\documentclass[a4paper, USenglish, cleveref, autoref, thm-restate]{lipics-v2021}

\nolinenumbers 
\hideLIPIcs

\usepackage{amsmath,amssymb,amsfonts, amsthm}
\usepackage{mathdots}
\usepackage{graphicx}
\usepackage{tikz}

\title{The local complexity of certifying parity}

\author{Nicolas Bousquet}{CNRS, INSA Lyon, UCBL, LIRIS, UMR5205, F-69622 Villeurbanne, France }{0000-0003-0170-0503}{}{}{}

\author{Laurent Feuilloley}{CNRS, INSA Lyon, UCBL, LIRIS, UMR5205, F-69622 Villeurbanne, France }{0000-0002-3994-0898}{}{}{}

\author{Jorge Valenzuela}{Universidad de Chile, Santiago, Chile}{0009-0008-3936-3227}{}{}{}

\author{Sébastien Zeitoun}{CNRS, INSA Lyon, UCBL, LIRIS, UMR5205, F-69622 Villeurbanne, France }{0009-0003-2675-8581}{}{}{}

\authorrunning{N. Bousquet, L. Feuilloley, J. Valenzuela, S. Zeitoun}

\keywords{Local certification, proof-labeling schemes, locally checkable proofs, space complexity, parity, congruence, lower bound}

\ccsdesc{Theory of computation~Distributed algorithms} 

\funding{This work is supported by the ANR grant ENEDISC (ANR-24-CE48-7768).}

\acknowledgements{The authors would like to thank Chilian-French ECOS-Sud-ANID program for funding the trip of the third author to France.}

\renewcommand{\P}{\mathcal{P}}
\renewcommand{\mod}{\text{~mod~}}
\newcommand{\C}{\mathcal{C}}

\newcommand{\N}{\mathbb{N}}
\newcommand{\T}{\mathcal{T}}
\newcommand{\dlog}{\text{{\normalshape dlog}}^\ast}
\newcommand{\tow}{\text{{\normalshape tow}}}

\newtheorem*{theorem*}{Theorem}
\newtheorem{question}{Question}
\newtheorem{open}{Open problem}

\begin{document}
\maketitle

\begin{abstract}
In this paper, we consider the problem of locally certifying that the size of a network is even, or more generally, congruent to some fixed number. 
The parity property is one of the simplest global properties, and it plays an intriguing role in local certification. 
On the one hand, it is one of the simplest properties in cycles because it is equivalent to 2-colorability, and hence can be certified with a single bit. 
On the other hand, in general graphs, no non-trivial lower bound on the size of the certificates is known, and the known upper bound basically consists in certifying the \emph{exact} value of $n$. 
In addition, the nature of the problem makes all the known lower bound approaches fail. 

We uncover a surprising landscape for parity across different models and graph structures:
\begin{itemize}
    \item In general graphs equipped with identifiers, when allowing verification radius 2, parity can be certified with a constant number of bits. 
    \item But in the model of anonymous graphs and allowing verification radius only 1, parity requires $\Omega(\log \log^*n)$ bits. 
    \item Finally, in bounded expansion graph classes (such as bounded-degree graphs and planar graphs), the lower bound does not apply: in the same restricted model we can design a constant-size certification.
\end{itemize}

We introduce several new tools that we expect to be useful in other contexts, in particular ways to \emph{encode a parent at each node with a constant number of bits} (via implicit use of the IDs and conflict-free colorings) and a new lower bound technique, with complex topologies and higher-order Ramsey-type arguments.
\end{abstract}

%\tableofcontents{}

\newpage{}

\section{Introduction}

This paper deals with local certification, a type of labeling that allows the verification of graph properties in a distributed manner. In local certification, an all-powerful \emph{prover} assigns a label to every node, referred to as the \emph{certificate} of the node. Then, each node inspects its neighborhood, including the certificates, and decides to accept or reject. Such a scheme is correct if: for any graph satisfying the property, there exists a certificate assignment that makes all nodes accept, and for any graph not satisfying the property, no such assignment exists.\footnote{For convenience, we define local certification for graph properties, though it can also apply to distributed data structures like spanning trees.}

Local certification, which originates from the study of fault-tolerance and self-stabilization, is now a well-established topic. We refer to \cite{Feuilloley21} for an introduction to the area and provide formal definitions in Section~\ref{sec:model}. 

In recent years, numerous graph properties have been studied through the lens of local certification, such as planarity, bounded treewidth, forbidden subgraphs, and connectivity. However, a very simple property has been overlooked: \emph{parity}, that is, the property of having an even number of nodes. In this paper, we study this property, and uncover its surprising and insightful behavior.  

Precisely, we are interested in the \emph{local complexity of parity}, which is the \emph{minimum number of bits per certificate needed to certify that the graph has an even number of nodes}. All our results apply to arbitrary congruence, that is, having $0 \mod k$ nodes, but we focus on parity for simplicity.

In order to have a more informed discussion, let us summarize what is known about local certification of parity.

\subsection{What is known about certifying parity}

It appears that (what we know about) the local complexity of certifying parity varies significantly depending on our assumptions about the graph topology. Let us discuss cycles, trees and general graphs. 

\paragraph*{In cycles}

In cycles, parity is equivalent to bipartiteness, which is of course not true in general graphs. Bipartiteness can easily be certified (in all graphs) by providing a 2-coloring as a certification, and this uses only 1-bit certificates. Hence, the local complexity of parity in cycles is exactly $1$. Certifying that the cycle has length $0 \mod k$, for some constant $k$, is also easy: the prover only needs to encode a counter modulo $k$ following an arbitrary orientation of the cycle,  and this uses $\log k$ bits. 
These certifications do not require the graph to have unique identifiers, and the nodes only need to communicate with their direct neighbors.

Interestingly, certifying that a cycle has an \emph{odd} number of nodes (or more generally, that its length is $a \mod k$, for fixed $a \neq 0$ and $k$) requires $\Omega(\log n)$ bits (where $n$ is the number of nodes) and unique identifiers, even if the nodes can communicate with neighbors at some arbitrary constant distance~\cite{GoosS16}. 
At an intuitive level, this difference between even and odd size can be justified in the following way. For parity, it is enough to use everywhere the same pattern of certificate, alternating 0s and 1s. Now for ``non-parity'', one could try to adapt this approach by using the same alternating pattern, except on exactly one node. But certifying that there exists exactly one node with a special role boils down to certifying the existence of a leader, which is known to require unique identifiers and local complexity $\Theta(\log n)$.

\paragraph*{In trees} 

The case of trees is essential for this paper, because all our upper bounds build on the certification in trees. 
Suppose first that we have a rooted tree, with the edges oriented towards the root. 
We can certify parity in such trees with just one bit. 
In correct instances, every node receives as its certificate the number of nodes in its subtree (including itself) modulo 2. We will refer to these counters as \emph{parity aggregates}.
Parity aggregates are easy to check in directed trees: every node verifies that the sum of its children certificates, plus 1, is equal to its certificate, modulo 2. 
This scheme can be adapted to undirected trees as follows. 
On correct instances, the prover chooses an arbitrary root, gives to every node its distance to the root modulo 3, in addition to the parity aggregates. It is well-known that such counters, that we will call \emph{mod-3 counters}, encode and certify a proper orientation in trees.  

Note that this orientation technique via mod-3 counters does not encode a proper spanning tree in general graphs. There are two main issues: (1) the counters might encode a directed graph with cycles and (2) a node might not be able to identify its parent: if it has counter $c$, there might be several neighbors with counter $c-1$.  

%\textcolor{red}{On cycles (and more generally to graphs) we can indeed also choose an orientation towards a root ans use the distance to the root modulo 3 technique to orient edges. However, the prover can fool the nodes for two reasons: (i) on treees, acyclicity ensures that the process stops (some vertex will not have any neighbor labeled with its label minus one) and that node might not exist on cycles and, (ii) on trees, every node has at most one parent which is not the case on a BFS tree.}

\paragraph*{In general graphs} 

While parity certification is well-understood in trees and cycles, no specific technique is known for parity in general graphs. 
Consequently, the state-of-the-art upper bound consists in certifying the value of $n$, and checking that it is even. 
In such a scheme, on correct instances the prover (1) certifies a spanning tree and (2) encodes at each node the size of its subtree. The nodes then (1) check the spanning tree, (2) check the subtree counters, and (3) the root checks that the value of its subtree is even. 
This scheme requires identifiers and uses $\Theta(\log n)$ bits even for just the spanning tree part, even when allowing communication at arbitrary constant distance~\cite{GoosS16}.
Regarding the lower bound, we know almost nothing: only that at least $1$ bit is required. This contrasts sharply with non-parity which requires $\Omega(\log n)$~\cite{GoosS16} in a powerful model and in cycles, as mentioned earlier. 

Given the wide gap between the trivial 1-bit lower bound, and the ``overkill'' upper bound, we ask:

\begin{question} 
What is the local complexity of parity in general graphs? Is it constant, or $O(\log n)$, or in between? Do we need identifiers to achieve it?  
\end{question}

%%%%%%%%%%%%%%%%%%%%%%%
\subsection{Motivations}
%%%%%%%%%%%%%%%%%%%%%%%%

Let us now describe our motivations to study the local certification of parity, beyond its simplicity and the gap between upper and lower bounds.

\paragraph*{A challenge for lower bound techniques}

Parity may be the simplest example of a property for which standard lower bound techniques fail. 
The two main lower bound techniques used in local certification are the so-called cut-and-plug and communication complexity techniques. 
Let us explain why the known implementations of these techniques fail, and why we believe that any sophisticated adaptation is likely to fail as well.
\begin{itemize}
    \item Cut-and-plug works by selecting accepting cycles, and using a counting argument to show that combining chunks of these cycles can create a \emph{no}-instance, where the nodes incorrectly accept~\cite{KormanKP10, GoosS16}. Since this method relies on cycles, and cycles are trivial for parity, the basic version of this technique cannot be applied. 
    While cut-and-plug has been applied to graphs more complex than cycles, \emph{e.g.} in \cite{FeuilloleyFMRRT21}, these graphs remain cycle-like and the certification for cycles can easily be adapted to them.
    In essence, to apply the counting arguments there must be a way to easily cut and combine pieces of \emph{yes}-instances, and this requirement  constrains the topology in a way that makes parity easy. 
   
    \item Communication complexity reductions typically take the following form. For every $n$, there exists a base graph, known to all nodes, with a fixed identifier assignment. This base graph contains two special subsets of vertices that are far apart in the graph, with a small cut between them. The lower bound instances are built by adding arbitrary edges within each subset. The intuition of the argument is then the following. In order to correctly decide the property, at least one node must know both special edge sets, which is a lot of information. Since this information can only be transferred via certificates, and the graph has a small cut, some node must have a large certificate.
    In such proofs, all the nodes know the base graph and thus the total number of nodes, which of course makes parity trivial.\footnote{Note that in general the a priori knowledge of the base graph is a strength and not a weakness of the technique, since it allows to derive lower bounds in a more powerful model. For parity it is simply too strong to be useful.} 

    One can of course imagine more flexible constructions, where the base graph is not fully known. However, we argue that the communication complexity argument is fundamentally inadequate for parity. This type of proof always relies on the fact that a lot of information should go through a small cut, forcing the certificates to be large. But for parity the only information needed is whether each part has an odd or even count! This basically rules out any strategy where one would only care about communication between a few well-identified parts of the graph. Thus, extensions to multi-party complexity are equally unpromising.  
\end{itemize}

%\textcolor{red}{N. Proposal to replace what's above. Congruence is maybe the simplest example of property where all the usual certification lower bound techniques fail. 
%The three main lower bound techniques are the so called cut-and-plug and communication complexity techniques~\cite{Feuilloley21} and lifts and covering maps.  \textcolor{red}{N. pourquoi la ref?}
%\\
%The cut-and-plug was initially designed for cycle topology \cite{KormanKP10, GoosS16}, which is useless for us, since cycle are easy for parity. Even if that proof technique has been adapted for more complex structures than cycles,  see \emph{e.g.} in \cite{FeuilloleyFMRRT21}. However, the spirit of the constructions are stil  based on cycle-like ideas and do not adapt for parity for that reason. 
%\\
%For communication complexity lower bounds, the core of the arguments is usually that a lot of information should go through a cut, and implicitly communication on each side of the cut is for free. But for parity this approach indeed fails: the only piece of information we need to transfer through the cut is whether each part is odd or even. For readers familiar with 2-party communication complexity: Alice can simply send to Bob whether she had an even number of vertices or not, and Bob can accept if and only he has the same count modulo $2$. This argument basiclaly rules out any strategy where one would only care about communication between to few parts of the network.} 

A coarse summary of the above is that lower bound instances that can be easily partitioned, especially via small cuts, are inadequate for proving lower bounds for parity. 
Another classic technique to prove lower bounds in distributed computing, one that does not require small cuts, is to use lifts or covering maps. 
Unfortunately, this approach also fails quickly and irreparably. Indeed, the size of a lift is always a multiple of the size of the original graph, hence the parity is preserved, and no contradiction can be derived.

\paragraph*{Congruence predicates as a class of properties}  

Recent work proved that in anonymous cycles, no property has local complexity both super-constant and sublogarithmic (that is, in $\omega(1) \cap o(\log n)$)~\cite{BousquetFZ25-DISC}. 
The properties that are in the $O(1)$-regime are basically those that can be described by a finite number of equations of the form $n = 0 \mod k$. 
Thus, parity is the canonical example of a well-motivated class of properties, 
and our work investigates the behavior of this class, when moving from cycles to general graphs. 

\paragraph*{A natural extension of the metatheorem line of research}

There is an ongoing effort to understand the local complexity of first-order (FO) and monadic second-order (MSO) formulas. These formulas capture locality in a logical framework, rather than relying on the bounded verification radius typically used in locally checkable labelings. See \cite{BlinFFGGMRT26} for a discussion of this approach.

Such properties can have a large local complexity in general graphs. For example, non-3-colorability is a simple FO formula that has complexity $\Omega(n^2/\log n)$ (almost matching the $O(n^2)$ upper bound that holds for any property). 
The main idea of the current line of research is that these properties become tractable in more structured graphs. Theorems establishing upper bounds for specific logics and structured graphs are called \emph{meta-theorems}, and a series of such theorems have emerged in recent years \cite{FeuilloleyBP22, BaterisnaC25, FraigniaudMRT23, FraigniaudMRT24, Cook0M25, BlinFFGGMRT26}. These theorems provide certifications of size $O(\log n)$ for all FO or all MSO formulas on graphs of bounded treedepth/pathwidth/treewidth/cliquewidth/expansion. 

Counting modulo a constant is a typical example of a property that \emph{cannot} be expressed in (plain) FO or MSO. In that sense, studying parity and similar properties complements the study of these logics. We will see that the meta-theorem approach also inspired us to study structured graphs.

%%%%%%%%%%%%%%%%%%%%%%%%%%%%%%%%
\subsection{Results and techniques}
%%%%%%%%%%%%%%%%%%%%%%%%%%%%%%%

Let us now state our three main results and briefly sketch their proofs. 

\paragraph*{Upper bound in general graphs}

First, we prove that \emph{parity can be certified with a constant number of bits, in a sufficiently powerful model of certification}. 

\begin{theorem*}[Informal version of Theorem~\ref{thm:upper-bound-general-graphs} in Section~\ref{sec:upper-bound-general-graphs}]
\label{thm:informal-general-graphs}
In general graphs equipped with identifiers, and verification radius 2, parity can be certified with $3$ bits.
\end{theorem*}

The technical challenge for proving this theorem is that the only approach we know for certifying parity (beyond cycles) is to aggregate counters along the edges of a spanning tree, and a spanning tree requires $\Omega(\log n)$ bits to be encoded. 
Intuitively, when encoding a spanning tree locally, we need to allow every node to identify its parent and children. Using identifiers explicitly is of course too expensive, if we aim for constant complexity. One approach is to encode new identifiers that are only locally unique, but this is not better in graphs that have linear maximum degree.  Another approach is to use the local structure of the graph to distinguish between nodes, but this does not work in highly symmetric graphs.

\medskip

We make several observations that allow us to bypass this obstacle. 
\begin{itemize}
\item We do not need the prover to be able to encode an adversarially chosen spanning tree, \emph{any} spanning tree is fine. 
\item The nodes can distinguish (some of) their neighbors, even without certificates: for example, their neighbor with the smallest ID. 
\item Using mod-3 counters from a chosen root, we can guide the direction in which the nodes should find their parent.  
\end{itemize}

Combining these observations, we design the following strategy for the prover:
\begin{enumerate}
    \item Choose an arbitrary root. 
    \item Provide every node with its distance to the root modulo 3. 
    \item Compute the spanning tree formed by every node choosing among its neighbors closer to the root the one with the smallest ID. 
    \item Provide every node with the size of its subtree in this spanning tree, modulo 2. 
\end{enumerate}

Two challenges arise. First, by communicating with its direct neighbors, a node can easily identify which node is its parent, but not which nodes are its children. This is why we need a verification radius of 2: every node $v$ can then check whether each of its neighbors has chosen $v$ as its parent. Given this, every node can check that the parity is aggregated correctly, as in the tree case. 

The second challenge is to make the scheme robust against an adversarial prover in \emph{no}-instances. 
A blatant weakness of the strategy above is that we cannot guarantee that the parent relation, defined by the mod-3 counter and smallest-ID, is indeed acyclic. 
For example, in a cycle of size $n=3k$, the prover can assign labels 1, 2, 3, 1, 2, 3...  and pretend that these are mod-3 counters. In this case, the identifiers are irrelevant for the verification, and no node detects a cycle in the parent relation. 
The pointers should then be considered as pointers to \emph{successors}, and not as pointers to parents.
In general, the successor relation globally defines a collection \emph{pseudo-trees}: cycles with trees attached to them.  Surprisingly, we can prove that \emph{this does not compromise the correctness of the scheme}! Indeed, we can prove that if all nodes accept with our scheme, then the pointers must form a collection of pseudo-trees of even size, implying that the graph is a correct instance.   

\smallskip

We now make a few observations about the theorem and the proof.

First, even though the proof is involved, the scheme itself is very simple and highly efficient in terms of certificate size: it only needs two bits for the mod-$3$ counters, and one bit for the parity aggregates. 

Second, while the identifiers are required, the certificates only depend on the identifier assignment in an order-invariant way: if another ID assignment with the same global order were used, the same certificate assignment would still work. 
As far as we know, this is the first example of an order-invariant local certification. In the literature, either the certificates do not depend on the IDs (\emph{e.g.} for certifying $k$-colorability) or they rely on the actual value, such as when naming a leader in the graph.

\paragraph*{Lower bound in general graphs}

We now investigate whether the identifiers and verification radius 2 are necessary to achieve constant local complexity. We present the following lower bound.

\begin{theorem*}[Informal version of Theorem~\ref{thm:explicit_lower_bound_parity} in Section~\ref{sec:lower-bound}]
In anonymous graphs, with verification radius 1, the local complexity of parity is in $\Omega(\log \log^*n)$. 
\end{theorem*}

As previously discussed, the standard lower bound techniques fail for parity. At a high level, we bypass the obstacles of these techniques by: (1) using more complex graph structures, (2) aiming at changing node degrees without causing rejection, and (3) replacing the classic counting argument by a higher-order Ramsey-type theorem.  

The general strategy is simple. We begin with a family of graphs, each with an even number of nodes. We show that, for any possible scheme using short certificates, and for any graph $G$ of this family, there exist two nodes $a$ and $b$, that have the following property. 
If we take two copies of $G$ and identify the vertex $a$ from the first copy with the vertex $b$ from the second copy, all the vertices accept the resulting graph.
This is a contradiction, because the new graph has $2|G|-1$ nodes, which is odd. 
We emphasize that in this reasoning we combine two instances by identifying vertices, rather than by rerouting edges as is usually done. 

The graph family we use consists of ``powerset graphs''. A graph $G$ in this family is a bipartite graph $G=(V_1 \cup V_2, E)$, where $V_1$ has $m$ nodes (with $m$ odd), and $V_2$ has $2^m-1$ nodes. For every non-empty subset of $V_1$, there exists exactly one node of $V_2$ that is linked to exactly this subset. 
See Figure~\ref{fig:powerset-graph}.

\begin{figure}[!h]
    \centering
    \scalebox{0.9}{
    \tikzset{every picture/.style={line width=0.75pt}}

\begin{tikzpicture}[x=0.75pt,y=0.75pt,yscale=-1,xscale=1]

% Box and label for V1 (top row)
\draw [line width=0.5pt, color=gray, rounded corners=5pt] (170,120) rectangle (400,152);
\node at (170,120) [above left, xshift=2pt, yshift=-2pt, color=black] {$V_1$};

% Circles for V1 (top row) - smaller and filled with black
\foreach \x in {194.02, 285.98, 374.02} {
    \draw [fill=black] (\x,135.98) circle (3);
}

% Box and label for V2 (bottom row)
\draw [line width=0.5pt, color=gray, rounded corners=5pt] (110,200) rectangle (460,232);
\node at (110,200) [above left, xshift=2pt, yshift=-2pt, color=black] {$V_2$};

% Circles for V2 (bottom row) - smaller and filled with black
\foreach \x in {134.02, 185.98, 235.98, 285.98, 337.95, 387.95, 435.98} {
    \draw [fill=black] (\x,215.98) circle (3);
}

% Lines from V1 (top) to V2 (bottom) - edges
\draw (194.02,135.98) -- (134.02,215.98);
\draw (194.02,135.98) -- (285.98,215.98);
\draw (194.02,135.98) -- (337.95,215.98);
\draw (194.02,135.98) -- (435.98,215.98);

\draw (285.98,135.98) -- (185.98,215.98);
\draw (285.98,135.98) -- (285.98,215.98);
\draw (285.98,135.98) -- (387.95,215.98);
\draw (285.98,135.98) -- (435.98,215.98);

\draw (374.02,135.98) -- (235.98,215.98);
\draw (374.02,135.98) -- (337.95,215.98);
\draw (374.02,135.98) -- (387.95,215.98);
\draw (374.02,135.98) -- (435.98,215.98);

\end{tikzpicture}}
    \caption{A graph $G$ from our lower bound instance family. Here $m=3$, $2^m -1=7$ and $n=10$.}
    \label{fig:powerset-graph}
\end{figure}
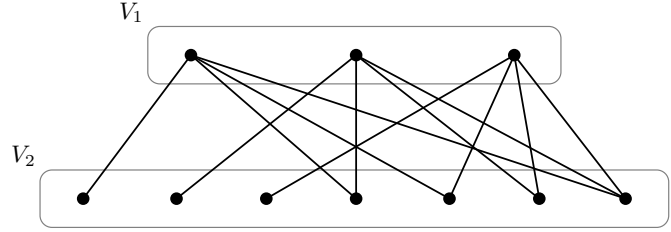

We use the \emph{finite unions theorem}~\cite[Chapter~3]{Graham91}, a powerful tool from Ramsey theory about colorings of set families, to prove the following. 
In our lower bound graphs, for any certificate assignment using $o(\log \log^*n)$ bits per node,  there exist three nodes $a$, $b$ and $c$ in $V_2$, such that: (1) they are given the same certificate, and (2) they are adjacent to sets $A$, $B$ and $C$ respectively, such that $A \cap B =\emptyset$ and $C=A\cup B$.\footnote{We can avoid using the finite union theorem, by relying on more standard Ramsey theorems and additive combinatorics results, but this only yields an $\omega(1)$ bound rather than $\Omega(\log \log^* n)$.}

Now, as announced, we take two copies of $G$, with the same accepting certificate assignment, and identify vertex $a$ from the first copy with vertex $b$ from the second copy. 
This certificate assignment is well-defined because $a$ and $b$ have the same certificate. It is also accepting: all the nodes have the same view as in $G$, except the merged vertex, which has exactly the same view as vertex $c$ in $G$, and therefore also accepts. Note that for this part of the argument to hold, we crucially need the verification radius to be one, and the network to be anonymous.

\paragraph*{Upper bound in structured sparse graphs}

Our lower bound graphs have complex structures, in the sense that many standard parameters are unbounded in this graph family, for example the maximum degree, the treewidth and even the VC dimension. 
We investigate whether this is an artifact of our proof, or whether it is unavoidable. A concrete question is: Can we still get super-constant lower bounds if we bound some parameters or restrict to well-known graph classes, such as planar graphs? 

We show that this complexity is necessary, by proving an upper bound in a very general family of structured graphs.  

\begin{theorem*}[Informal version of Theorem~\ref{thm:UB_parity_bounded_expansion} in Section~\ref{sec:upper-bounded-expansion}]
In any family of anonymous graphs of bounded expansion, parity has constant local complexity, even with verification radius 1.
\end{theorem*}

Bounded expansion basically ensures that a graph class is \emph{structurally} sparse. In contrast, the standard notion of sparsity is coarse, as it only requires a global condition: a bounded ratio of edges to vertices. 
Under this definition, even dense graphs can be made sparse by simply attaching a sufficiently long path, or subdividing their edges. 
In many cases, this basic definition of sparsity is not demanding enough to be useful, and bounded expansion is an essential alternative.
Intuitively, a family of graphs has bounded expansion if it is sparse and cannot be made arbitrarily dense by contracting small subgraphs. 
A bit more formally, a family of graphs has bounded expansion if, for any $t$, the density remains bounded by some function $f(t)$, after contracting subgraphs of radius at most $t$ and deleting vertices. See Section~\ref{sec:upper-bounded-expansion} for proper definitions.

Bounded expansion graphs encompass many standard graph classes such as bounded degree graphs, random graphs of bounded average degree (with high probability), bounded treewidth graphs, planar graphs, and more generally minor-free graphs. 
These classes have proven interesting from a local certification perspective: for example, there exists constant-size certifications for perfect matchings in planar graphs~\cite{BousquetFZ25-SIDMA} and leader election in grids~\cite{FeuilloleySS26, ChalopinCK26}. 
A recent breakthrough establishes an efficient meta-theorem in bounded expansion graphs in the CONGEST model, and the same tools enable local certification of all FO formulas with local complexity $O(\log n)$~\cite{BlinFFGGMRT26}. This result is incomparable with ours, since we use constant-size certificates and parity is not in FO. Nevertheless it highlights that this structural notion of sparsity is central for local certification, and distributed graph algorithms in general. 

\medskip

We turn to our proof techniques. 
We build on the ideas behind the scheme for general graphs. 
Recall that our upper bound relied on two key assumptions: unique identifiers (to identify a successor) and verification radius 2 (to identify the predecessor/children). We prove that we can leverage the structure of bounded expansion graph classes to certify parity without these assumptions. 

The main tool is \emph{conflict-free colorings}. These are colorings of the nodes of the graph, such that for every node, some color is used exactly once in its neighborhood. That is, for every node $v$, there exists a neighbor $w$ with some color $c$, such that no other neighbor of $v$ has color $c$. 
It was proved recently that for bounded expansion graph classes we can get conflict-free colorings with a constant number of colors~\cite{Hickingbotham23}. 
The basic idea is as follows: the prover computes a conflict-free coloring of the graph, and provides each node with its color, and the color of its successor. This is unambiguous because there is a unique node with this color in the neighborhood. 
This alone is insufficient to ensure correct parity aggregates, but we can enhance it by using mod-3 counters, and a collection of conflict-free colorings taking the mod-3 layers into account.
The proof then follows the general case: we prove that a graph is accepted if and only if it has a spanning collection of even pseudo-trees, which is equivalent to parity. 

Although conflict-free colorings have been extensively studied in graph theory, we are unaware of any prior application in distributed computing.

%%%%%%%%%%%%%%%%%%%%%%%
\subsection{Discussions}
%%%%%%%%%%%%%%%%%%%%%%%%

\paragraph*{A surprising discrepancy in the complexities}

Our first two theorems show that in two rather close models of local certification the local complexity can be different. We consider the two models to be close because they follow the standard framework, unlike more recent generalizations using multiple provers or a non-constant number communication rounds. We now elaborate on why this discrepancy is surprising. 

We now describe how the local complexity depends on the verification radius and the identifier model in previous work.
Regarding the verification radius, two behaviors have been identified.
\begin{enumerate}
    \item For most classic properties, increasing the verification radius from $1$ to an arbitrary constant only decreases the certificate size by a constant factor. This phenomenon was studied in~\cite{GoosS16} by introducing locally checkable proofs, as a generalization of the original model of proof-labeling schemes~\cite{KormanKP10}.\footnote{It is conjectured that the certificate size always decreases at least linearly with the verification radius. This known as the trade-off conjecture, raised several times in the literature, and a slightly weaker version proved very recently~\cite{FiltserF26}.}
    \item For some properties, increasing the verification radius eventually makes certification trivial, and this threshold can be deduced directly from the definition of the property. For example, certifying that a graph is triangle-free requires near-linear certificates at distance $1$, but is trivial at distance 2. See~\cite{BousquetCFPZ24} for a detailed study of such problems.
\end{enumerate}

We also observe a dichotomy for the role of the identifiers. 
\begin{enumerate}
    \item Most classic properties, such as leader election or clique detection, cannot be solved without identifiers for any constant verification radius.
    \item For some properties, the natural certification does not use identifiers, for example $k$-colorability.
\end{enumerate}

With parity, we encounter in a new situation: a change of verification radius and identifier model alters the local  complexity, for no obvious reason. Indeed, it is not clear why parity exhibits a fundamental difference between radius 1 and radius 2, and why identifiers would be required, given that the problem does not rely on a node having a specific role. 

\paragraph*{The problems we leave open}

Since our first two theorems differ in both the verification radius and the identifier model, a natural question arises: what about the intermediate models?  

\begin{open}
    Is it possible to certify parity with constant local complexity, in anonymous graphs with verification radius 2, or in identified graphs with verification radius 1?
\end{open}

We note that the simplest adaptations of our techniques do not work: subdividing edges in the lower bound breaks the constructions, and replacing IDs with pseudo-IDs is challenging without assumptions on the structure of the graph.  

Even within the models of this paper, large gaps in our understanding persist. 
The local complexity in the anonymous case with radius $1$ remains wide open. 
An upper bound in this setting can be obtained by using a conflict-free coloring (as in our third contribution), but this requires $O(\log n)$ bits in general graphs, as up to $n$ colors are needed. 
Since our lower bound is just above constant, the gap remains very large. 

\begin{open}
    What is the local complexity of parity in anonymous graphs with verification radius 1?
\end{open}

It is possible that an $\Omega(\log n)$ lower bound can be achieved, but this would require fundamentally new ideas. Indeed, in our lower bound, we crucially need some nodes of different degree to be assigned the same certificate, and this cannot be guaranteed in our lower bound instances if the certificates have size $\Omega(\log \log n)$. This is because the nodes that are on the side $V_2$ (where there is one node for every subset of side $V_1$) 
have degree $m = O(\log n)$, and writing the degree of each such node in its certificate only takes $O(\log \log n)$ bits. 

\paragraph*{On the power of constant-size local certification}

Let us now shift our focus slightly and discuss the power of constant-size certification in anonymous graphs with verification radius 1. 
It is known that the languages of paths (with inputs) that have constant local complexity in this setting are exactly the regular languages~\cite{BousquetFZ25-DISC}. 
And since counting modulo a constant is straightforward for a finite automaton, our lower bound implies that when moving from paths to general graphs, we lose the ability to count.
We wonder whether the properties with constant local complexity in general graphs correspond to restrictions of regular languages. 
For example, we could explore analogues of locally testable languages (regular languages defined by forbidden subwords), or star-free (regular languages that can be expressed without Kleene star). However this approach is challenging: the typical example of a property of local complexity is 3-colorability, and 3-colorable graphs can have highly complex structures. 

\begin{open}
    Can we characterize the properties that have constant local complexity without referring to local certification? In general graphs, or even in specific graph classes? 
\end{open}

\paragraph*{Similarities with LCL theory}

Finally, readers familiar with the theory of locally checkable labelings will recognize similarities between some concepts of this paper and the $O(\log^*n)$ regime of complexity in LOCAL. 
First, for the upper bound in general graphs we use a scheme that is order-invariant, which is reminiscent of Naor-Stockmeyer result~\cite{NaorS95} which shows that any constant-time algorithm in the LOCAL model can be made order-invariant. 
Second, the $\log^*n$ complexity of our lower bound is the same as the complexity of symmetry breaking (\emph{e.g.} 3-coloring a cycle). 
We do not know whether a deep connection exists between these two branches of distributed computing, but we highlight that the graphs used for the lower bound \emph{do not have bounded degree}, which implies that easy adaptations are not expected. 

\subsection{Organization of the paper}

The rest of the paper is organized as follows. In Section~\ref{sec:model}, we provide the formal definitions of the models used in the paper. 
In Section~\ref{sec:upper-bound-general-graphs} and~\ref{sec:upper-bounded-expansion}, we state and prove our upper bounds for general graphs and bounded expansion graph classes, respectively. 
In Section~\ref{sec:lower-bound}, we state and prove our lower bound. Note that the order of the sections differs from the order in which discussed the results in the introduction. 

%%%%%%%%%%%%%%%%%%%%%%%%%%%%
\section{Model and definitions}
\label{sec:model}
%%%%%%%%%%%%%%%%%%%%%%%%%%%%

\paragraph*{Graph notions}
All the graphs we consider in this paper are simple, loopless, finite, undirected, and connected. If~$G$ is a graph, we denote by~$V(G)$ and~$E(G)$ its sets of vertices and edges respectively; we also use~$V$ and~$E$ when $G$ is clear from the context.
We denote by~$n$ the number of vertices. 

\paragraph*{Models of local certification.} Let $G = (V,E)$ be an $n$-vertex graph, and $c > 0$. 
An \emph{identifier assignment} of~$G$ in the range~$[n^c]$ is an \emph{injective} function $\mbox{ID} : V(G) \to [n^c]$. If~$G$ is equipped with an identifier assignment, we say that we are in the \emph{locally checkable proof model} (or \emph{LCP model} for short), else we say that we are in the \emph{anonymous model}. Finally, let~$C$ be a non-empty set, called the set of \emph{certificates}. A~\emph{certificate assignment} of~$G$ with certificates in~$C$ is a mapping~$c : V \to C$.

\paragraph*{Local view and verification algorithm} Let $G=(V,E)$ be a graph, $r \geqslant 1$, and $c$ be a certificate assignment for~$G$. Let $u \in V$. The \emph{local view of~$u$ at distance~$r$} consists in all the information available at distance at most~$r$ from~$u$ in~$G$, that is:

\begin{itemize}
	\item its own certificate and its own identifier (if any);
	\item all the vertices at distance at most~$r$ from~$u$, and all the edges having at least one endpoint at distance at most~$r-1$ from~$u$;
	\item the restriction of the certificate assignment~$c$ and of the identifier assignment (if any) to these vertices.
\end{itemize}	

A \emph{verification algorithm} (or \emph{verification procedure}) is a function taking as input the local view at distance~$r$ of a vertex, and outputting a decision, \emph{accept} or \emph{reject}. The integer~$r$ is called the~\emph{verification radius} of the algorithm. If the verification radius is equal to~$r$, we also say that we are in the~\emph{$r$-local model}.

\paragraph*{Certification scheme}
Let $\C$ be a class of graphs and $\P$ be a property on graphs in~$\C$.
Let $s : \N \to \N$. We say that there exists a \emph{certification scheme for~$\P$} with \emph{local complexity}~$s$ if there exists a verification algorithm such that the two following conditions are satisfied:

\begin{itemize}
	\item \emph{Completeness}: For every $n$-vertex graph~$G \in \C$ that satisfies~$\P$, and for every identifier assignment of $G$ (if we are in the locally checkable proof model), there exists a certificate assignment in $[2^{s(n)}]$ such that the verification of every vertex accepts (we say that the graph is \emph{globally accepted}).
	\item \emph{Soundness}: For every graph~$G \in \C$ that does not satisfy~$\P$, for every identifier assignment of~$G$ (if we are in the locally checkable proof model), for every $k \in \N$ and every certificate assignment in $\{0, \ldots, 2^{k}-1\}$, at least one vertex rejects.
\end{itemize}

We often say that the certificates are assigned by an external \emph{prover}.
Let us emphasize that, if $G$ does not satisfy~$\P$, then for any assignment of certificates \emph{of any size}, at least one vertex rejects.
Let us also point out the fact that in a certification scheme for a property~$\P$ in some class~$\C$, the vertices have the promise that the graph belongs to~$\C$. In other words, the certification scheme depends on the property~$\P$ \emph{and} on the class~$\C$, and we are not concerned by the output of the verification procedure of the vertices in graphs that do not belong to~$\C$.

\paragraph*{Computational power of the prover and vertices}
Note that, in the previous definitions, there is no limitation on the computational power of the prover and the vertices. Namely, the verification of the vertices is a just a function, with no more requirements (in particular, it is not necessarily computable in polynomial time, or even decidable).
This definition only makes stronger the lower bounds on the local complexity, by showing that they do not come from computational limits. However, when designing certification schemes in practice, it will often turn out to be computable in polynomial time.

\paragraph*{Discussion on the models} In general, it is easier to certify a property in the locally checkable proof model than in the anonymous model. Indeed, proving an upper bound on the local complexity of some property in the anonymous model also proves it in the LCP model, simply because a certification scheme in the anonymous model also works in the LCP model. Similarly, a lower bound in the LCP model automatically gives a lower bound in the anonymous model.

%%%%%%%%%%%%%%%%%%%%%%%
\section{Upper bound in general graphs}
\label{sec:upper-bound-general-graphs}
%%%%%%%%%%%%%%%%%%%%%%%

In this section, we establish the constant local complexity of (generalizations of) parity in the model with verification radius 2 and identifiers. Let us start by a couple of definitions, and then state and prove our theorem. 

\begin{definition}
	The \emph{underlying graph} of an oriented graph~$G$ is the \mbox{non-oriented} graph obtained from~$G$ by forgetting the orientation of the edges.
\end{definition}

\begin{definition}
	An~\emph{oriented cycle} is an oriented graph whose underlying graph is a cycle, and such that each vertex has exactly one incoming and one outgoing edge.
\end{definition}

\begin{theorem}
\label{thm:upper-bound-general-graphs}
	Let $k \geqslant 2$, and let $\P$ be the property of having a number of vertices divisible by~$k$. In the $2$-local LCP model, the local complexity of~$\P$ is~$O(\log k)$.
\end{theorem}

\begin{proof}
	Let us describe a $2$-local certification scheme for~$\P$, in the locally checkable proof model.
	Let~$G$ be a graph. Let us first describe how the prover assigns the certificates on a correct instance. First, it chooses an arbitrary vertex~$r \in V(G)$, and computes the following spanning tree~$\T$ rooted at~$r$: for every vertex~$u \neq r$ at distance~$d$ from~$r$, its parent in~$\T$ is its neighbor at distance~$d-1$ with minimum identifier. Note that this defines indeed a spanning tree (because~$\T$ is a particular breadth-first-search tree).
	Finally, the prover writes, in the certificate of each vertex~$u$, the following information (which can be encoded on~$O(\log k)$ bits):
    
	\begin{enumerate}[(C1)]
		\item the distance~$d_u$ from~$u$ to~$r$ modulo~$3$; \label{item:C1_parity}
		\item the number of vertices~$n_u$ in the subtree of~$\T$ rooted at~$u$ modulo~$k$. \label{item:C2_parity}
	\end{enumerate}
	
	Let us now describe the verification procedure of the vertices. %For every vertex \mbox{$u \in V(G)$}, we denote respectively by~$ u$ and~$n_u$ the values of the distance modulo~3 from~$u$ to~$r$, and of the number of vertices modulo~$k$ in the subtree rooted at~$u$, written by the prover in the certificate of~$u$.
	Consider a certificate assignment~$c$ to the vertices (that is, some assignment of values~$d_u$ and~$n_u$).
	If the set~$\{v \in N(u) \; | \; d_v = d_u-1 \mod 3\}$ is empty, we say that~$u$ is a \emph{root} (with respect to the certificate assignment~$c$). If this set is non-empty, we define the~\emph{parent} of~$u$ (again, with respect to the certificate assignment~$c$) as the vertex with minimum identifier in this set. If~$v$ is the parent of~$u$, we also say that~$u$ is a~\emph{child} of~$v$. If~$u$ has no child, we say that it is a~\emph{leaf}. Note that, since each vertex~$u$ has a view at distance~$2$, it can determine, for every neighbor~$v$, if it is itself the parent of~$v$ (because it sees the certificates and identifiers of the neighbors of~$v$). Thus, it can determine its set of children. The verification of every vertex~$u$ consists in the following, by rejecting if at least one condition is not satisfied, and accepting if both conditions are satisfied:

	\begin{enumerate}[(V1)]
		\item $u$ checks that $n_u = 1 + \sum_{v \in S} n_v \mod k$, where~$S$ denotes its set of children; \label{item:V1_parity}
		\item if $u$ is a root, it checks that~$n_u = 0 \mod k$. \label{item:V2_parity}
	\end{enumerate}
	
	The completeness of this scheme is straightforward: indeed, if the number of vertices in~$G$ is divisible by~$k$, and if the prover assigns the certificates as described in~\ref{item:C1_parity} and~\ref{item:C2_parity}, no vertex will reject in its verification procedure.
	
	Let us prove the soundness of this scheme.
	Let~$c$ be a certificate assignment such that no vertex rejects, and let us prove that~$G$ satisfies~$\P$. In fact, we will not prove that the parent relationship with respect to~$c$ defines a correct spanning tree, because it might not be the case. However, even if this parent relationship does not define a correct spanning tree, we will still be able to prove that the property~$\P$ holds.
	Let us consider the following oriented graph~$\overrightarrow{G}$, on the same set of vertices as~$G$. For every edge~$\{u,v\}$ of~$G$, we put the edge~$u \rightarrow v$ in~$\overrightarrow{G}$ if and only if $v$ is the parent of~$u$ with respect to~$c$. Note that, in~$\overrightarrow{G}$, each vertex has at most one outgoing edge, and it does have one if and only if it is not a root.
	Let us consider a connected component~$C$ of the underlying graph of~$\overrightarrow{G}$.
	Let us denote by~$|C|$ its number of vertices. We aim to prove that~$|C|$ is divisible by~$k$.
	Since each vertex of~$\overrightarrow{G}$ has at most one outgoing edge, then~$C$ has at most~$|C|$ edges in~$\overrightarrow{G}$. Moreover, since~$C$ is connected in the underlying graph of~$\overrightarrow{G}$, it has at least~$|C|-1$ edges. We say that~$C$ is a~\emph{type-$1$ component} if it has~$|C|-1$ edges, and a~\emph{type-$2$ component} if it has~$|C|$ edges (see Figure~\ref{fig:parity_types} for an illustration). We now treat these two cases separately.
	
	\begin{figure}[h!]
		\centering
		
		\begin{tikzpicture}[x=0.55pt,y=0.55pt,yscale=-1,xscale=1]
			%uncomment if require: \path (0,456); %set diagram left start at 0, and has height of 456
			
			%Straight Lines [id:da6798167179245248] 
			\draw    (382.57,126.77) -- (377.27,169.72) ;
			\draw [shift={(376.91,172.69)}, rotate = 277.04] [fill={rgb, 255:red, 0; green, 0; blue, 0 }  ]   [draw opacity=0] (8.93,-4.29) -- (0,0) -- (8.93,4.29) -- cycle    ;
			%Straight Lines [id:da28344864220342136] 
			\draw    (375.82,178.39) -- (414.98,208.61) ;
			\draw [shift={(417.35,210.45)}, rotate = 217.66] [fill={rgb, 255:red, 0; green, 0; blue, 0 }  ]   [draw opacity=0] (8.93,-4.29) -- (0,0) -- (8.93,4.29) -- cycle    ;
			%Straight Lines [id:da4134068988300049] 
			\draw    (421.29,214.71) -- (464.41,203.44) ;
			\draw [shift={(467.31,202.68)}, rotate = 165.36] [fill={rgb, 255:red, 0; green, 0; blue, 0 }  ]   [draw opacity=0] (8.93,-4.29) -- (0,0) -- (8.93,4.29) -- cycle    ;
			%Straight Lines [id:da8220560335231119] 
			\draw    (472.95,201.34) -- (479.49,150.53) ;
			\draw [shift={(479.88,147.55)}, rotate = 97.34] [fill={rgb, 255:red, 0; green, 0; blue, 0 }  ]   [draw opacity=0] (8.93,-4.29) -- (0,0) -- (8.93,4.29) -- cycle    ;
			%Straight Lines [id:da44873072247261103] 
			\draw    (480.61,141.8) -- (444.77,110.65) ;
			\draw [shift={(442.51,108.68)}, rotate = 40.99] [fill={rgb, 255:red, 0; green, 0; blue, 0 }  ]   [draw opacity=0] (8.93,-4.29) -- (0,0) -- (8.93,4.29) -- cycle    ;
			%Straight Lines [id:da9974837322267381] 
			\draw    (438.4,104.58) -- (390.87,123.84) ;
			\draw [shift={(388.08,124.97)}, rotate = 337.95] [fill={rgb, 255:red, 0; green, 0; blue, 0 }  ]   [draw opacity=0] (8.93,-4.29) -- (0,0) -- (8.93,4.29) -- cycle    ;
			%Straight Lines [id:da4327526530726231] 
			\draw    (446.37,61.01) -- (440.08,95.94) ;
			\draw [shift={(439.55,98.9)}, rotate = 280.21] [fill={rgb, 255:red, 0; green, 0; blue, 0 }  ]   [draw opacity=0] (8.93,-4.29) -- (0,0) -- (8.93,4.29) -- cycle    ;
			%Straight Lines [id:da8869214984913414] 
			\draw    (508.69,246.07) -- (478.56,208.11) ;
			\draw [shift={(476.7,205.76)}, rotate = 51.56] [fill={rgb, 255:red, 0; green, 0; blue, 0 }  ]   [draw opacity=0] (8.93,-4.29) -- (0,0) -- (8.93,4.29) -- cycle    ;
			%Straight Lines [id:da09600167191956466] 
			\draw    (531.37,210.01) -- (481.68,202.47) ;
			\draw [shift={(478.71,202.02)}, rotate = 8.62] [fill={rgb, 255:red, 0; green, 0; blue, 0 }  ]   [draw opacity=0] (8.93,-4.29) -- (0,0) -- (8.93,4.29) -- cycle    ;
			%Straight Lines [id:da40730684076733426] 
			\draw    (291.37,238.01) -- (323.23,212.13) ;
			\draw [shift={(325.56,210.24)}, rotate = 140.91] [fill={rgb, 255:red, 0; green, 0; blue, 0 }  ]   [draw opacity=0] (8.93,-4.29) -- (0,0) -- (8.93,4.29) -- cycle    ;
			%Straight Lines [id:da9847178993896005] 
			\draw    (330.37,207.01) -- (368.1,182.59) ;
			\draw [shift={(370.62,180.96)}, rotate = 147.09] [fill={rgb, 255:red, 0; green, 0; blue, 0 }  ]   [draw opacity=0] (8.93,-4.29) -- (0,0) -- (8.93,4.29) -- cycle    ;
			%Straight Lines [id:da3679500627557788] 
			\draw    (503.37,87.01) -- (484.29,133.81) ;
			\draw [shift={(483.16,136.59)}, rotate = 292.18] [fill={rgb, 255:red, 0; green, 0; blue, 0 }  ]   [draw opacity=0] (8.93,-4.29) -- (0,0) -- (8.93,4.29) -- cycle    ;
			%Straight Lines [id:da7083119473788027] 
			\draw    (531.37,126.01) -- (489.02,139.19) ;
			\draw [shift={(486.15,140.08)}, rotate = 342.72] [fill={rgb, 255:red, 0; green, 0; blue, 0 }  ]   [draw opacity=0] (8.93,-4.29) -- (0,0) -- (8.93,4.29) -- cycle    ;
			%Straight Lines [id:da3429382712405462] 
			\draw    (554.37,85.01) -- (535.28,118.13) ;
			\draw [shift={(533.78,120.73)}, rotate = 299.97] [fill={rgb, 255:red, 0; green, 0; blue, 0 }  ]   [draw opacity=0] (8.93,-4.29) -- (0,0) -- (8.93,4.29) -- cycle    ;
			%Straight Lines [id:da02209065431156576] 
			\draw    (595.37,57.01) -- (561.44,79.77) ;
			\draw [shift={(558.95,81.44)}, rotate = 326.14] [fill={rgb, 255:red, 0; green, 0; blue, 0 }  ]   [draw opacity=0] (8.93,-4.29) -- (0,0) -- (8.93,4.29) -- cycle    ;
			%Straight Lines [id:da9345853851964616] 
			\draw    (578.37,126.01) -- (540.17,126.17) ;
			\draw [shift={(537.17,126.19)}, rotate = 359.76] [fill={rgb, 255:red, 0; green, 0; blue, 0 }  ]   [draw opacity=0] (8.93,-4.29) -- (0,0) -- (8.93,4.29) -- cycle    ;
			%Straight Lines [id:da9813490790491334] 
			\draw    (631.37,94.01) -- (585.59,121) ;
			\draw [shift={(583.01,122.52)}, rotate = 329.48] [fill={rgb, 255:red, 0; green, 0; blue, 0 }  ]   [draw opacity=0] (8.93,-4.29) -- (0,0) -- (8.93,4.29) -- cycle    ;
			%Straight Lines [id:da6815114148380811] 
			\draw    (638.37,138.01) -- (586.99,127.82) ;
			\draw [shift={(584.04,127.23)}, rotate = 11.22] [fill={rgb, 255:red, 0; green, 0; blue, 0 }  ]   [draw opacity=0] (8.93,-4.29) -- (0,0) -- (8.93,4.29) -- cycle    ;
			%Straight Lines [id:da8575134919656696] 
			\draw    (610.37,180.01) -- (582.75,133.64) ;
			\draw [shift={(581.21,131.07)}, rotate = 59.21] [fill={rgb, 255:red, 0; green, 0; blue, 0 }  ]   [draw opacity=0] (8.93,-4.29) -- (0,0) -- (8.93,4.29) -- cycle    ;
			%Straight Lines [id:da4153447874434719] 
			\draw    (362.37,84.01) -- (378.89,118.78) ;
			\draw [shift={(380.18,121.49)}, rotate = 244.59] [fill={rgb, 255:red, 0; green, 0; blue, 0 }  ]   [draw opacity=0] (8.93,-4.29) -- (0,0) -- (8.93,4.29) -- cycle    ;
			%Straight Lines [id:da04803474810147468] 
			\draw    (315.37,66.01) -- (354.52,80.11) ;
			\draw [shift={(357.34,81.13)}, rotate = 199.81] [fill={rgb, 255:red, 0; green, 0; blue, 0 }  ]   [draw opacity=0] (8.93,-4.29) -- (0,0) -- (8.93,4.29) -- cycle    ;
			%Straight Lines [id:da4005398418921863] 
			\draw    (376.41,38.16) -- (364.98,75.6) ;
			\draw [shift={(364.1,78.47)}, rotate = 286.98] [fill={rgb, 255:red, 0; green, 0; blue, 0 }  ]   [draw opacity=0] (8.93,-4.29) -- (0,0) -- (8.93,4.29) -- cycle    ;
			%Straight Lines [id:da808083477201513] 
			\draw    (82.41,54.16) -- (164.58,47.26) ;
			\draw [shift={(167.57,47.01)}, rotate = 175.2] [fill={rgb, 255:red, 0; green, 0; blue, 0 }  ]   [draw opacity=0] (8.93,-4.29) -- (0,0) -- (8.93,4.29) -- cycle    ;
			%Straight Lines [id:da8867446646427095] 
			\draw    (173.37,47.01) -- (158.42,65.87) -- (142.25,86.26) ;
			\draw [shift={(140.39,88.61)}, rotate = 308.41] [fill={rgb, 255:red, 0; green, 0; blue, 0 }  ]   [draw opacity=0] (8.93,-4.29) -- (0,0) -- (8.93,4.29) -- cycle    ;
			%Straight Lines [id:da4397543742094139] 
			\draw    (137.4,93.58) -- (182.95,121.48) ;
			\draw [shift={(185.51,123.05)}, rotate = 211.49] [fill={rgb, 255:red, 0; green, 0; blue, 0 }  ]   [draw opacity=0] (8.93,-4.29) -- (0,0) -- (8.93,4.29) -- cycle    ;
			%Straight Lines [id:da9873007157193421] 
			\draw    (74.57,132.77) -- (130.73,99.23) ;
			\draw [shift={(133.3,97.69)}, rotate = 149.15] [fill={rgb, 255:red, 0; green, 0; blue, 0 }  ]   [draw opacity=0] (8.93,-4.29) -- (0,0) -- (8.93,4.29) -- cycle    ;
			%Straight Lines [id:da16953651256677138] 
			\draw    (190.61,125.8) -- (142.37,152.77) ;
			\draw [shift={(139.75,154.23)}, rotate = 330.79] [fill={rgb, 255:red, 0; green, 0; blue, 0 }  ]   [draw opacity=0] (8.93,-4.29) -- (0,0) -- (8.93,4.29) -- cycle    ;
			%Straight Lines [id:da5860714274915128] 
			\draw    (127.29,230.71) -- (86.18,204.21) ;
			\draw [shift={(83.66,202.58)}, rotate = 32.81] [fill={rgb, 255:red, 0; green, 0; blue, 0 }  ]   [draw opacity=0] (8.93,-4.29) -- (0,0) -- (8.93,4.29) -- cycle    ;
			%Straight Lines [id:da7123642096758855] 
			\draw    (78.82,199.39) -- (128.01,162.78) ;
			\draw [shift={(130.42,160.99)}, rotate = 143.34] [fill={rgb, 255:red, 0; green, 0; blue, 0 }  ]   [draw opacity=0] (8.93,-4.29) -- (0,0) -- (8.93,4.29) -- cycle    ;
			%Straight Lines [id:da5817474868565166] 
			\draw    (181.95,202.34) -- (141.31,162.87) ;
			\draw [shift={(139.15,160.78)}, rotate = 44.16] [fill={rgb, 255:red, 0; green, 0; blue, 0 }  ]   [draw opacity=0] (8.93,-4.29) -- (0,0) -- (8.93,4.29) -- cycle    ;
			%Shape: Circle [id:dp7487502823366963] 
			\draw  [fill={rgb, 255:red, 255; green, 255; blue, 255 }  ,fill opacity=1 ]   (141.51,97.68) .. controls (139.24,99.95) and (135.57,99.95) .. (133.3,97.69) .. controls (131.04,95.42) and (131.03,91.75) .. (133.3,89.48) .. controls (135.56,87.22) and (139.23,87.21) .. (141.5,89.48) .. controls (143.77,91.74) and (143.77,95.41) .. (141.51,97.68) -- cycle ;
			%Shape: Circle [id:dp915387144576147] 
			\draw  [fill={rgb, 255:red, 255; green, 255; blue, 255 }  ,fill opacity=1 ]   (82.01,194.55) .. controls (84.69,196.31) and (85.42,199.91) .. (83.66,202.58) .. controls (81.89,205.26) and (78.3,206) .. (75.62,204.23) .. controls (72.95,202.47) and (72.21,198.87) .. (73.98,196.2) .. controls (75.74,193.52) and (79.34,192.79) .. (82.01,194.55) -- cycle ;
			%Shape: Circle [id:dp7375697218192767] 
			\draw  [fill={rgb, 255:red, 255; green, 255; blue, 255 }  ,fill opacity=1 ]   (176.15,202.34) .. controls (176.15,199.13) and (178.75,196.54) .. (181.95,196.54) .. controls (185.15,196.54) and (187.75,199.13) .. (187.75,202.34) .. controls (187.75,205.54) and (185.15,208.14) .. (181.95,208.14) .. controls (178.75,208.14) and (176.15,205.54) .. (176.15,202.34) -- cycle ;
			%Shape: Circle [id:dp4357018276781581] 
			\draw  [fill={rgb, 255:red, 255; green, 255; blue, 255 }  ,fill opacity=1 ]   (167.57,47.01) .. controls (167.57,43.81) and (170.17,41.21) .. (173.37,41.21) .. controls (176.58,41.21) and (179.17,43.81) .. (179.17,47.01) .. controls (179.17,50.21) and (176.58,52.81) .. (173.37,52.81) .. controls (170.17,52.81) and (167.57,50.21) .. (167.57,47.01) -- cycle ;
			%Shape: Circle [id:dp4338198011903639] 
			\draw  [fill={rgb, 255:red, 255; green, 255; blue, 255 }  ,fill opacity=1 ]   (138.4,152.61) .. controls (140.86,154.66) and (141.2,158.31) .. (139.15,160.78) .. controls (137.11,163.24) and (133.45,163.58) .. (130.99,161.53) .. controls (128.52,159.49) and (128.19,155.83) .. (130.23,153.37) .. controls (132.28,150.9) and (135.94,150.56) .. (138.4,152.61) -- cycle ;
			%Shape: Circle [id:dp8951366396426678] 
			\draw  [fill={rgb, 255:red, 255; green, 255; blue, 255 }  ,fill opacity=1 ]   (187.86,130.9) .. controls (185.04,129.38) and (183.99,125.87) .. (185.51,123.05) .. controls (187.03,120.23) and (190.55,119.17) .. (193.37,120.69) .. controls (196.19,122.21) and (197.24,125.73) .. (195.72,128.55) .. controls (194.2,131.37) and (190.68,132.42) .. (187.86,130.9) -- cycle ;
			%Shape: Circle [id:dp35752610344301883] 
			\draw  [fill={rgb, 255:red, 255; green, 255; blue, 255 }  ,fill opacity=1 ]   (121.49,230.71) .. controls (121.49,227.5) and (124.08,224.91) .. (127.29,224.91) .. controls (130.49,224.91) and (133.09,227.5) .. (133.09,230.71) .. controls (133.09,233.91) and (130.49,236.51) .. (127.29,236.51) .. controls (124.08,236.51) and (121.49,233.91) .. (121.49,230.71) -- cycle ;
			%Shape: Circle [id:dp8611234272879673] 
			\draw  [fill={rgb, 255:red, 255; green, 255; blue, 255 }  ,fill opacity=1 ]   (68.77,132.77) .. controls (68.77,129.57) and (71.37,126.97) .. (74.57,126.97) .. controls (77.78,126.97) and (80.37,129.57) .. (80.37,132.77) .. controls (80.37,135.98) and (77.78,138.57) .. (74.57,138.57) .. controls (71.37,138.57) and (68.77,135.98) .. (68.77,132.77) -- cycle ;
			%Shape: Circle [id:dp5279440364215211] 
			\draw  [fill={rgb, 255:red, 255; green, 255; blue, 255 }  ,fill opacity=1 ]   (76.61,54.16) .. controls (76.61,50.95) and (79.21,48.36) .. (82.41,48.36) .. controls (85.62,48.36) and (88.21,50.95) .. (88.21,54.16) .. controls (88.21,57.36) and (85.62,59.96) .. (82.41,59.96) .. controls (79.21,59.96) and (76.61,57.36) .. (76.61,54.16) -- cycle ;
			%Shape: Circle [id:dp9241070683663125] 
			\draw  [fill={rgb, 255:red, 255; green, 255; blue, 255 }  ,fill opacity=1 ]   (432.72,103.43) .. controls (433.35,100.29) and (436.41,98.26) .. (439.55,98.9) .. controls (442.69,99.53) and (444.72,102.59) .. (444.09,105.73) .. controls (443.45,108.87) and (440.39,110.9) .. (437.25,110.27) .. controls (434.11,109.63) and (432.08,106.57) .. (432.72,103.43) -- cycle ;
			%Shape: Circle [id:dp8391891247085108] 
			\draw  [fill={rgb, 255:red, 255; green, 255; blue, 255 }  ,fill opacity=1 ]   (378.39,183.59) .. controls (375.52,185.01) and (372.04,183.83) .. (370.62,180.96) .. controls (369.2,178.09) and (370.38,174.61) .. (373.25,173.19) .. controls (376.12,171.77) and (379.6,172.95) .. (381.02,175.82) .. controls (382.44,178.69) and (381.26,182.17) .. (378.39,183.59) -- cycle ;
			%Shape: Circle [id:dp23400541637751393] 
			\draw  [fill={rgb, 255:red, 255; green, 255; blue, 255 }  ,fill opacity=1 ]   (473.63,195.58) .. controls (476.82,195.96) and (479.09,198.84) .. (478.71,202.02) .. controls (478.33,205.2) and (475.45,207.47) .. (472.27,207.1) .. controls (469.08,206.72) and (466.81,203.83) .. (467.19,200.65) .. controls (467.57,197.47) and (470.45,195.2) .. (473.63,195.58) -- cycle ;
			%Shape: Circle [id:dp7345097062789326] 
			\draw  [fill={rgb, 255:red, 255; green, 255; blue, 255 }  ,fill opacity=1 ]   (440.57,61.01) .. controls (440.57,57.81) and (443.17,55.21) .. (446.37,55.21) .. controls (449.58,55.21) and (452.17,57.81) .. (452.17,61.01) .. controls (452.17,64.21) and (449.58,66.81) .. (446.37,66.81) .. controls (443.17,66.81) and (440.57,64.21) .. (440.57,61.01) -- cycle ;
			%Shape: Circle [id:dp3185163366338807] 
			\draw  [fill={rgb, 255:red, 255; green, 255; blue, 255 }  ,fill opacity=1 ]   (512.4,241.61) .. controls (514.86,243.66) and (515.2,247.31) .. (513.15,249.78) .. controls (511.11,252.24) and (507.45,252.58) .. (504.99,250.53) .. controls (502.52,248.49) and (502.19,244.83) .. (504.23,242.37) .. controls (506.28,239.9) and (509.94,239.56) .. (512.4,241.61) -- cycle ;
			%Shape: Circle [id:dp58579450129731] 
			\draw  [fill={rgb, 255:red, 255; green, 255; blue, 255 }  ,fill opacity=1 ]   (478.89,136.26) .. controls (481.95,135.31) and (485.2,137.02) .. (486.15,140.08) .. controls (487.1,143.14) and (485.39,146.39) .. (482.34,147.34) .. controls (479.28,148.29) and (476.03,146.58) .. (475.07,143.52) .. controls (474.12,140.46) and (475.83,137.21) .. (478.89,136.26) -- cycle ;
			%Shape: Circle [id:dp7079086909193586] 
			\draw  [fill={rgb, 255:red, 255; green, 255; blue, 255 }  ,fill opacity=1 ]   (417.03,218.64) .. controls (414.85,216.29) and (415,212.62) .. (417.35,210.45) .. controls (419.71,208.27) and (423.37,208.42) .. (425.55,210.77) .. controls (427.72,213.13) and (427.57,216.8) .. (425.22,218.97) .. controls (422.87,221.14) and (419.2,221) .. (417.03,218.64) -- cycle ;
			%Shape: Circle [id:dp5508817691944637] 
			\draw  [fill={rgb, 255:red, 255; green, 255; blue, 255 }  ,fill opacity=1 ]   (377.29,129.17) .. controls (375.97,126.25) and (377.26,122.81) .. (380.18,121.49) .. controls (383.1,120.17) and (386.53,121.46) .. (387.86,124.38) .. controls (389.18,127.3) and (387.88,130.73) .. (384.97,132.06) .. controls (382.05,133.38) and (378.61,132.08) .. (377.29,129.17) -- cycle ;
			%Shape: Circle [id:dp4158405041274127] 
			\draw  [fill={rgb, 255:red, 255; green, 255; blue, 255 }  ,fill opacity=1 ]   (370.61,38.16) .. controls (370.61,34.95) and (373.21,32.36) .. (376.41,32.36) .. controls (379.62,32.36) and (382.21,34.95) .. (382.21,38.16) .. controls (382.21,41.36) and (379.62,43.96) .. (376.41,43.96) .. controls (373.21,43.96) and (370.61,41.36) .. (370.61,38.16) -- cycle ;
			%Shape: Circle [id:dp2573987619267686] 
			\draw  [fill={rgb, 255:red, 255; green, 255; blue, 255 }  ,fill opacity=1 ]   (356.84,82.28) .. controls (357.79,79.23) and (361.04,77.52) .. (364.1,78.47) .. controls (367.16,79.43) and (368.86,82.68) .. (367.91,85.74) .. controls (366.96,88.79) and (363.71,90.5) .. (360.65,89.55) .. controls (357.59,88.59) and (355.88,85.34) .. (356.84,82.28) -- cycle ;
			%Shape: Circle [id:dp9814724015409597] 
			\draw  [fill={rgb, 255:red, 255; green, 255; blue, 255 }  ,fill opacity=1 ]   (309.57,66.01) .. controls (309.57,62.81) and (312.17,60.21) .. (315.37,60.21) .. controls (318.58,60.21) and (321.17,62.81) .. (321.17,66.01) .. controls (321.17,69.21) and (318.58,71.81) .. (315.37,71.81) .. controls (312.17,71.81) and (309.57,69.21) .. (309.57,66.01) -- cycle ;
			%Shape: Circle [id:dp6826200424964559] 
			\draw  [fill={rgb, 255:red, 255; green, 255; blue, 255 }  ,fill opacity=1 ]   (525.57,210.01) .. controls (525.57,206.81) and (528.17,204.21) .. (531.37,204.21) .. controls (534.58,204.21) and (537.17,206.81) .. (537.17,210.01) .. controls (537.17,213.21) and (534.58,215.81) .. (531.37,215.81) .. controls (528.17,215.81) and (525.57,213.21) .. (525.57,210.01) -- cycle ;
			%Shape: Circle [id:dp8487989536317069] 
			\draw  [fill={rgb, 255:red, 255; green, 255; blue, 255 }  ,fill opacity=1 ]   (333.6,211.83) .. controls (330.94,213.61) and (327.34,212.9) .. (325.56,210.24) .. controls (323.77,207.58) and (324.48,203.98) .. (327.14,202.19) .. controls (329.81,200.41) and (333.41,201.12) .. (335.19,203.78) .. controls (336.98,206.44) and (336.26,210.04) .. (333.6,211.83) -- cycle ;
			%Shape: Circle [id:dp021893369678809504] 
			\draw  [fill={rgb, 255:red, 255; green, 255; blue, 255 }  ,fill opacity=1 ]   (285.57,238.01) .. controls (285.57,234.81) and (288.17,232.21) .. (291.37,232.21) .. controls (294.58,232.21) and (297.17,234.81) .. (297.17,238.01) .. controls (297.17,241.21) and (294.58,243.81) .. (291.37,243.81) .. controls (288.17,243.81) and (285.57,241.21) .. (285.57,238.01) -- cycle ;
			%Shape: Circle [id:dp9951131941094259] 
			\draw  [fill={rgb, 255:red, 255; green, 255; blue, 255 }  ,fill opacity=1 ]   (531.55,120.21) .. controls (534.75,120.31) and (537.27,122.98) .. (537.17,126.19) .. controls (537.07,129.39) and (534.4,131.9) .. (531.2,131.81) .. controls (528,131.71) and (525.48,129.04) .. (525.58,125.83) .. controls (525.67,122.63) and (528.35,120.11) .. (531.55,120.21) -- cycle ;
			%Shape: Circle [id:dp34767680410011237] 
			\draw  [fill={rgb, 255:red, 255; green, 255; blue, 255 }  ,fill opacity=1 ]   (550.81,80.43) .. controls (553.34,78.47) and (556.98,78.92) .. (558.95,81.44) .. controls (560.92,83.97) and (560.47,87.62) .. (557.94,89.58) .. controls (555.41,91.55) and (551.77,91.1) .. (549.8,88.57) .. controls (547.83,86.05) and (548.28,82.4) .. (550.81,80.43) -- cycle ;
			%Shape: Circle [id:dp7395284573791436] 
			\draw  [fill={rgb, 255:red, 255; green, 255; blue, 255 }  ,fill opacity=1 ]   (583.43,123.17) .. controls (585,125.96) and (584.01,129.5) .. (581.21,131.07) .. controls (578.42,132.64) and (574.88,131.64) .. (573.32,128.85) .. controls (571.75,126.06) and (572.74,122.52) .. (575.54,120.95) .. controls (578.33,119.38) and (581.86,120.38) .. (583.43,123.17) -- cycle ;
			%Shape: Circle [id:dp7547251267110703] 
			\draw  [fill={rgb, 255:red, 255; green, 255; blue, 255 }  ,fill opacity=1 ]   (637.02,143.65) .. controls (633.91,142.9) and (631.99,139.77) .. (632.73,136.66) .. controls (633.48,133.54) and (636.61,131.62) .. (639.73,132.37) .. controls (642.84,133.12) and (644.76,136.25) .. (644.01,139.36) .. controls (643.27,142.48) and (640.14,144.4) .. (637.02,143.65) -- cycle ;
			%Shape: Circle [id:dp8289384789323447] 
			\draw  [fill={rgb, 255:red, 255; green, 255; blue, 255 }  ,fill opacity=1 ]   (633.69,99.33) .. controls (630.75,100.61) and (627.33,99.26) .. (626.06,96.32) .. controls (624.78,93.39) and (626.12,89.97) .. (629.06,88.69) .. controls (632,87.41) and (635.41,88.76) .. (636.69,91.69) .. controls (637.97,94.63) and (636.63,98.05) .. (633.69,99.33) -- cycle ;
			%Shape: Circle [id:dp0807156101282428] 
			\draw  [fill={rgb, 255:red, 255; green, 255; blue, 255 }  ,fill opacity=1 ]   (589.57,57.01) .. controls (589.57,53.81) and (592.17,51.21) .. (595.37,51.21) .. controls (598.58,51.21) and (601.17,53.81) .. (601.17,57.01) .. controls (601.17,60.21) and (598.58,62.81) .. (595.37,62.81) .. controls (592.17,62.81) and (589.57,60.21) .. (589.57,57.01) -- cycle ;
			%Shape: Circle [id:dp048159756164447676] 
			\draw  [fill={rgb, 255:red, 255; green, 255; blue, 255 }  ,fill opacity=1 ]   (497.57,87.01) .. controls (497.57,83.81) and (500.17,81.21) .. (503.37,81.21) .. controls (506.58,81.21) and (509.17,83.81) .. (509.17,87.01) .. controls (509.17,90.21) and (506.58,92.81) .. (503.37,92.81) .. controls (500.17,92.81) and (497.57,90.21) .. (497.57,87.01) -- cycle ;
			%Shape: Circle [id:dp21344418239641216] 
			\draw  [fill={rgb, 255:red, 255; green, 255; blue, 255 }  ,fill opacity=1 ]   (605.18,182.58) .. controls (603.76,179.71) and (604.93,176.23) .. (607.8,174.81) .. controls (610.68,173.39) and (614.15,174.57) .. (615.57,177.44) .. controls (616.99,180.31) and (615.82,183.79) .. (612.95,185.21) .. controls (610.07,186.63) and (606.59,185.45) .. (605.18,182.58) -- cycle ;

		\end{tikzpicture}

		\caption{Left: a type-$1$ component in~$\overrightarrow{G}$. Right: a type-$2$ component in~$\overrightarrow{G}$.}
		\label{fig:parity_types}
	\end{figure}
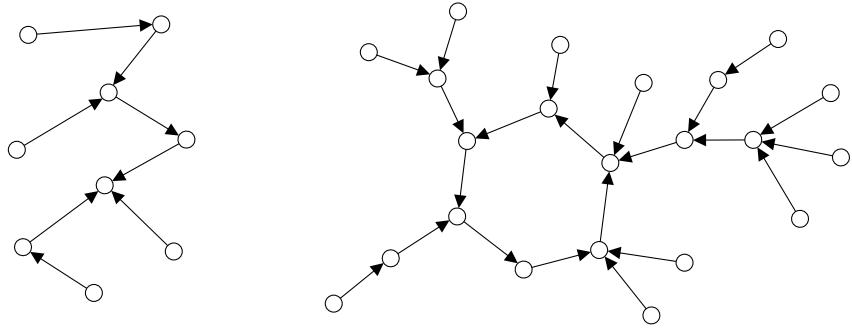
	
	First, assume that~$C$ is a type-$1$ component. Since it has~$|C|-1$ edges, it has a unique root, and its underlying graph is a tree. Thus, $C$ is a rooted tree. Let us denote its root by~$r$. Since no vertex rejects in~\ref{item:V1_parity}, for all vertex $u$ which is a leaf, we have~$n_u = 1$. Then, by an immediate induction, we obtain that for every vertex~$u$ in~$C$, $n_u$ is equal to the number of vertices in the subtree of~$C$ rooted at~$u$ modulo~$k$ (again because~$u$ does not reject in~\ref{item:V1_parity}). In particular, we have $n_r = |C| \mod k$. Finally, since~$r$ does not reject in~\ref{item:V2_parity}, we have $n_r = 0 \mod k$, so $|C|$ is divisible by~$k$.
	
	Now, assume that~$C$ is a type-$2$ component. Since it has~$|C|$ edges, its underlying graph is a tree plus an edge. In other words, it is a cycle plus some trees attached to the vertices of this cycle. Let us denote this cycle by~$C_0$. Since every vertex in~$\overrightarrow{G}$ has at most one outgoing neighbor, then $C_0$ is an oriented cycle in~$\overrightarrow{G}$. Let us denote by~$s$ the number of vertices in~$C_0$, and by~$u_1, \ldots, u_s$ its consecutive vertices. For every $i \in \{1, \ldots, s\}$, let us denote by~$T_i$ the tree attached to~$u_i$ (including~$u_i$ itself), and by~$t_i$ its number of vertices. The tree~$T_i$ is a rooted tree, and its root is~$u_i$ (because the outgoing neighbor of~$u_i$ is in~$C_0$, so~$u_i$ does not have any outgoing neighbor in~$T_i$). Also, we have $|C| = \sum_{i=1}^{s} t_i$. See Figure~\ref{fig:parity_type2} for an illustration.
	
	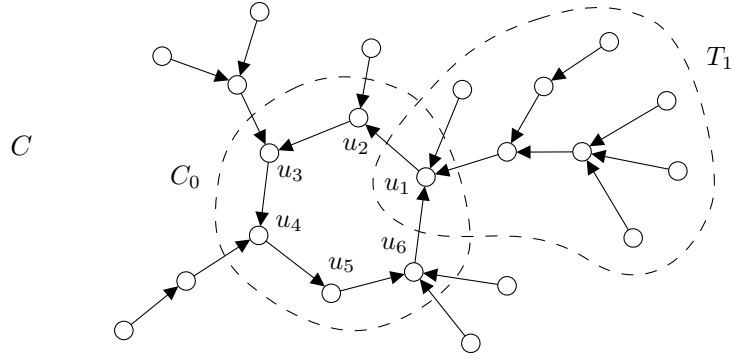
\begin{figure}[h!]
		\centering
		
		\begin{tikzpicture}[x=0.6pt,y=0.6pt,yscale=-1,xscale=1]
			
			%Straight Lines [id:da7815229050338993] 
			\draw    (253.81,122.77) -- (248.51,165.72) ;
			\draw [shift={(248.14,168.69)}, rotate = 277.04] [fill={rgb, 255:red, 0; green, 0; blue, 0 }  ]   [draw opacity=0] (8.93,-4.29) -- (0,0) -- (8.93,4.29) -- cycle    ;
			%Straight Lines [id:da8154880604332175] 
			\draw    (247.05,174.39) -- (286.21,204.61) ;
			\draw [shift={(288.59,206.45)}, rotate = 217.66] [fill={rgb, 255:red, 0; green, 0; blue, 0 }  ]   [draw opacity=0] (8.93,-4.29) -- (0,0) -- (8.93,4.29) -- cycle    ;
			%Straight Lines [id:da6124210010170225] 
			\draw    (292.52,210.71) -- (335.64,199.44) ;
			\draw [shift={(338.54,198.68)}, rotate = 165.36] [fill={rgb, 255:red, 0; green, 0; blue, 0 }  ]   [draw opacity=0] (8.93,-4.29) -- (0,0) -- (8.93,4.29) -- cycle    ;
			%Straight Lines [id:da13764285090779627] 
			\draw    (344.18,197.34) -- (350.73,146.53) ;
			\draw [shift={(351.11,143.55)}, rotate = 97.34] [fill={rgb, 255:red, 0; green, 0; blue, 0 }  ]   [draw opacity=0] (8.93,-4.29) -- (0,0) -- (8.93,4.29) -- cycle    ;
			%Straight Lines [id:da726985281289674] 
			\draw    (351.85,137.8) -- (316,106.65) ;
			\draw [shift={(313.74,104.68)}, rotate = 40.99] [fill={rgb, 255:red, 0; green, 0; blue, 0 }  ]   [draw opacity=0] (8.93,-4.29) -- (0,0) -- (8.93,4.29) -- cycle    ;
			%Straight Lines [id:da6332391292715687] 
			\draw    (309.64,100.58) -- (262.1,119.84) ;
			\draw [shift={(259.32,120.97)}, rotate = 337.95] [fill={rgb, 255:red, 0; green, 0; blue, 0 }  ]   [draw opacity=0] (8.93,-4.29) -- (0,0) -- (8.93,4.29) -- cycle    ;
			%Straight Lines [id:da45191643139053406] 
			\draw    (317.61,57.01) -- (311.32,91.94) ;
			\draw [shift={(310.78,94.9)}, rotate = 280.21] [fill={rgb, 255:red, 0; green, 0; blue, 0 }  ]   [draw opacity=0] (8.93,-4.29) -- (0,0) -- (8.93,4.29) -- cycle    ;
			%Straight Lines [id:da8907304129905909] 
			\draw    (379.93,242.07) -- (349.8,204.11) ;
			\draw [shift={(347.93,201.76)}, rotate = 51.56] [fill={rgb, 255:red, 0; green, 0; blue, 0 }  ]   [draw opacity=0] (8.93,-4.29) -- (0,0) -- (8.93,4.29) -- cycle    ;
			%Straight Lines [id:da7225269503731314] 
			\draw    (402.61,206.01) -- (352.91,198.47) ;
			\draw [shift={(349.94,198.02)}, rotate = 8.62] [fill={rgb, 255:red, 0; green, 0; blue, 0 }  ]   [draw opacity=0] (8.93,-4.29) -- (0,0) -- (8.93,4.29) -- cycle    ;
			%Straight Lines [id:da03733105153169314] 
			\draw    (162.61,234.01) -- (194.46,208.13) ;
			\draw [shift={(196.79,206.24)}, rotate = 140.91] [fill={rgb, 255:red, 0; green, 0; blue, 0 }  ]   [draw opacity=0] (8.93,-4.29) -- (0,0) -- (8.93,4.29) -- cycle    ;
			%Straight Lines [id:da08744229807665804] 
			\draw    (201.61,203.01) -- (239.33,178.59) ;
			\draw [shift={(241.85,176.96)}, rotate = 147.09] [fill={rgb, 255:red, 0; green, 0; blue, 0 }  ]   [draw opacity=0] (8.93,-4.29) -- (0,0) -- (8.93,4.29) -- cycle    ;
			%Straight Lines [id:da5584199460807312] 
			\draw    (374.61,83.01) -- (355.52,129.81) ;
			\draw [shift={(354.39,132.59)}, rotate = 292.18] [fill={rgb, 255:red, 0; green, 0; blue, 0 }  ]   [draw opacity=0] (8.93,-4.29) -- (0,0) -- (8.93,4.29) -- cycle    ;
			%Straight Lines [id:da3124094565885608] 
			\draw    (402.61,122.01) -- (360.25,135.19) ;
			\draw [shift={(357.39,136.08)}, rotate = 342.72] [fill={rgb, 255:red, 0; green, 0; blue, 0 }  ]   [draw opacity=0] (8.93,-4.29) -- (0,0) -- (8.93,4.29) -- cycle    ;
			%Straight Lines [id:da7201361694437378] 
			\draw    (425.61,81.01) -- (406.51,114.13) ;
			\draw [shift={(405.01,116.73)}, rotate = 299.97] [fill={rgb, 255:red, 0; green, 0; blue, 0 }  ]   [draw opacity=0] (8.93,-4.29) -- (0,0) -- (8.93,4.29) -- cycle    ;
			%Straight Lines [id:da005517249698063886] 
			\draw    (466.61,53.01) -- (432.67,75.77) ;
			\draw [shift={(430.18,77.44)}, rotate = 326.14] [fill={rgb, 255:red, 0; green, 0; blue, 0 }  ]   [draw opacity=0] (8.93,-4.29) -- (0,0) -- (8.93,4.29) -- cycle    ;
			%Straight Lines [id:da25257092951711346] 
			\draw    (449.61,122.01) -- (411.41,122.17) ;
			\draw [shift={(408.41,122.19)}, rotate = 359.76] [fill={rgb, 255:red, 0; green, 0; blue, 0 }  ]   [draw opacity=0] (8.93,-4.29) -- (0,0) -- (8.93,4.29) -- cycle    ;
			%Straight Lines [id:da19989120497610557] 
			\draw    (502.61,90.01) -- (456.83,117) ;
			\draw [shift={(454.24,118.52)}, rotate = 329.48] [fill={rgb, 255:red, 0; green, 0; blue, 0 }  ]   [draw opacity=0] (8.93,-4.29) -- (0,0) -- (8.93,4.29) -- cycle    ;
			%Straight Lines [id:da46430162263478925] 
			\draw    (509.61,134.01) -- (458.22,123.82) ;
			\draw [shift={(455.28,123.23)}, rotate = 11.22] [fill={rgb, 255:red, 0; green, 0; blue, 0 }  ]   [draw opacity=0] (8.93,-4.29) -- (0,0) -- (8.93,4.29) -- cycle    ;
			%Straight Lines [id:da30304692022135615] 
			\draw    (481.61,176.01) -- (453.98,129.64) ;
			\draw [shift={(452.45,127.07)}, rotate = 59.21] [fill={rgb, 255:red, 0; green, 0; blue, 0 }  ]   [draw opacity=0] (8.93,-4.29) -- (0,0) -- (8.93,4.29) -- cycle    ;
			%Straight Lines [id:da5351238031245863] 
			\draw    (233.61,80.01) -- (250.12,114.78) ;
			\draw [shift={(251.41,117.49)}, rotate = 244.59] [fill={rgb, 255:red, 0; green, 0; blue, 0 }  ]   [draw opacity=0] (8.93,-4.29) -- (0,0) -- (8.93,4.29) -- cycle    ;
			%Straight Lines [id:da04566295641975837] 
			\draw    (186.61,62.01) -- (225.75,76.11) ;
			\draw [shift={(228.57,77.13)}, rotate = 199.81] [fill={rgb, 255:red, 0; green, 0; blue, 0 }  ]   [draw opacity=0] (8.93,-4.29) -- (0,0) -- (8.93,4.29) -- cycle    ;
			%Straight Lines [id:da30928104592286654] 
			\draw    (247.65,34.16) -- (236.21,71.6) ;
			\draw [shift={(235.33,74.47)}, rotate = 286.98] [fill={rgb, 255:red, 0; green, 0; blue, 0 }  ]   [draw opacity=0] (8.93,-4.29) -- (0,0) -- (8.93,4.29) -- cycle    ;
			%Shape: Circle [id:dp36359423047565054] 
			\draw  [fill={rgb, 255:red, 255; green, 255; blue, 255 }  ,fill opacity=1 ]   (303.95,99.43) .. controls (304.58,96.29) and (307.64,94.26) .. (310.78,94.9) .. controls (313.92,95.53) and (315.95,98.59) .. (315.32,101.73) .. controls (314.69,104.87) and (311.63,106.9) .. (308.49,106.27) .. controls (305.35,105.63) and (303.32,102.57) .. (303.95,99.43) -- cycle ;
			%Shape: Circle [id:dp34952908519560044] 
			\draw  [fill={rgb, 255:red, 255; green, 255; blue, 255 }  ,fill opacity=1 ]   (249.62,179.59) .. controls (246.75,181.01) and (243.27,179.83) .. (241.85,176.96) .. controls (240.43,174.09) and (241.61,170.61) .. (244.48,169.19) .. controls (247.35,167.77) and (250.83,168.95) .. (252.25,171.82) .. controls (253.67,174.69) and (252.49,178.17) .. (249.62,179.59) -- cycle ;
			%Shape: Circle [id:dp26456844520647993] 
			\draw  [fill={rgb, 255:red, 255; green, 255; blue, 255 }  ,fill opacity=1 ]   (344.87,191.58) .. controls (348.05,191.96) and (350.32,194.84) .. (349.94,198.02) .. controls (349.56,201.2) and (346.68,203.47) .. (343.5,203.1) .. controls (340.32,202.72) and (338.05,199.83) .. (338.42,196.65) .. controls (338.8,193.47) and (341.69,191.2) .. (344.87,191.58) -- cycle ;
			%Shape: Circle [id:dp4002112343416173] 
			\draw  [fill={rgb, 255:red, 255; green, 255; blue, 255 }  ,fill opacity=1 ]   (311.81,57.01) .. controls (311.81,53.81) and (314.4,51.21) .. (317.61,51.21) .. controls (320.81,51.21) and (323.41,53.81) .. (323.41,57.01) .. controls (323.41,60.21) and (320.81,62.81) .. (317.61,62.81) .. controls (314.4,62.81) and (311.81,60.21) .. (311.81,57.01) -- cycle ;
			%Shape: Circle [id:dp11892751632988607] 
			\draw  [fill={rgb, 255:red, 255; green, 255; blue, 255 }  ,fill opacity=1 ]   (383.63,237.61) .. controls (386.1,239.66) and (386.44,243.31) .. (384.39,245.78) .. controls (382.34,248.24) and (378.68,248.58) .. (376.22,246.53) .. controls (373.76,244.49) and (373.42,240.83) .. (375.47,238.37) .. controls (377.51,235.9) and (381.17,235.56) .. (383.63,237.61) -- cycle ;
			%Shape: Circle [id:dp2695492923788625] 
			\draw  [fill={rgb, 255:red, 255; green, 255; blue, 255 }  ,fill opacity=1 ]   (350.12,132.26) .. controls (353.18,131.31) and (356.43,133.02) .. (357.39,136.08) .. controls (358.34,139.14) and (356.63,142.39) .. (353.57,143.34) .. controls (350.51,144.29) and (347.26,142.58) .. (346.31,139.52) .. controls (345.36,136.46) and (347.07,133.21) .. (350.12,132.26) -- cycle ;
			%Shape: Circle [id:dp34088582470290407] 
			\draw  [fill={rgb, 255:red, 255; green, 255; blue, 255 }  ,fill opacity=1 ]   (288.26,214.64) .. controls (286.09,212.29) and (286.23,208.62) .. (288.59,206.45) .. controls (290.94,204.27) and (294.61,204.42) .. (296.78,206.77) .. controls (298.95,209.13) and (298.81,212.8) .. (296.45,214.97) .. controls (294.1,217.14) and (290.43,217) .. (288.26,214.64) -- cycle ;
			%Shape: Circle [id:dp9165402149739377] 
			\draw  [fill={rgb, 255:red, 255; green, 255; blue, 255 }  ,fill opacity=1 ]   (248.52,125.17) .. controls (247.2,122.25) and (248.49,118.81) .. (251.41,117.49) .. controls (254.33,116.17) and (257.77,117.46) .. (259.09,120.38) .. controls (260.41,123.3) and (259.12,126.73) .. (256.2,128.06) .. controls (253.28,129.38) and (249.85,128.08) .. (248.52,125.17) -- cycle ;
			%Shape: Circle [id:dp01093290274475578] 
			\draw  [fill={rgb, 255:red, 255; green, 255; blue, 255 }  ,fill opacity=1 ]   (241.85,34.16) .. controls (241.85,30.95) and (244.44,28.36) .. (247.65,28.36) .. controls (250.85,28.36) and (253.45,30.95) .. (253.45,34.16) .. controls (253.45,37.36) and (250.85,39.96) .. (247.65,39.96) .. controls (244.44,39.96) and (241.85,37.36) .. (241.85,34.16) -- cycle ;
			%Shape: Circle [id:dp49794346475180995] 
			\draw  [fill={rgb, 255:red, 255; green, 255; blue, 255 }  ,fill opacity=1 ]   (228.07,78.28) .. controls (229.02,75.23) and (232.28,73.52) .. (235.33,74.47) .. controls (238.39,75.43) and (240.1,78.68) .. (239.15,81.74) .. controls (238.19,84.79) and (234.94,86.5) .. (231.88,85.55) .. controls (228.82,84.59) and (227.12,81.34) .. (228.07,78.28) -- cycle ;
			%Shape: Circle [id:dp647527235077089] 
			\draw  [fill={rgb, 255:red, 255; green, 255; blue, 255 }  ,fill opacity=1 ]   (180.81,62.01) .. controls (180.81,58.81) and (183.4,56.21) .. (186.61,56.21) .. controls (189.81,56.21) and (192.41,58.81) .. (192.41,62.01) .. controls (192.41,65.21) and (189.81,67.81) .. (186.61,67.81) .. controls (183.4,67.81) and (180.81,65.21) .. (180.81,62.01) -- cycle ;
			%Shape: Circle [id:dp5532273395356848] 
			\draw  [fill={rgb, 255:red, 255; green, 255; blue, 255 }  ,fill opacity=1 ]   (396.81,206.01) .. controls (396.81,202.81) and (399.4,200.21) .. (402.61,200.21) .. controls (405.81,200.21) and (408.41,202.81) .. (408.41,206.01) .. controls (408.41,209.21) and (405.81,211.81) .. (402.61,211.81) .. controls (399.4,211.81) and (396.81,209.21) .. (396.81,206.01) -- cycle ;
			%Shape: Circle [id:dp11608227116407] 
			\draw  [fill={rgb, 255:red, 255; green, 255; blue, 255 }  ,fill opacity=1 ]   (204.84,207.83) .. controls (202.18,209.61) and (198.57,208.9) .. (196.79,206.24) .. controls (195.01,203.58) and (195.72,199.98) .. (198.38,198.19) .. controls (201.04,196.41) and (204.64,197.12) .. (206.43,199.78) .. controls (208.21,202.44) and (207.5,206.04) .. (204.84,207.83) -- cycle ;
			%Shape: Circle [id:dp7879917805034963] 
			\draw  [fill={rgb, 255:red, 255; green, 255; blue, 255 }  ,fill opacity=1 ]   (156.81,234.01) .. controls (156.81,230.81) and (159.4,228.21) .. (162.61,228.21) .. controls (165.81,228.21) and (168.41,230.81) .. (168.41,234.01) .. controls (168.41,237.21) and (165.81,239.81) .. (162.61,239.81) .. controls (159.4,239.81) and (156.81,237.21) .. (156.81,234.01) -- cycle ;
			%Shape: Circle [id:dp9345916066288514] 
			\draw  [fill={rgb, 255:red, 255; green, 255; blue, 255 }  ,fill opacity=1 ]   (402.78,116.21) .. controls (405.99,116.31) and (408.5,118.98) .. (408.41,122.19) .. controls (408.31,125.39) and (405.63,127.9) .. (402.43,127.81) .. controls (399.23,127.71) and (396.71,125.04) .. (396.81,121.83) .. controls (396.91,118.63) and (399.58,116.11) .. (402.78,116.21) -- cycle ;
			%Shape: Circle [id:dp117211329646147] 
			\draw  [fill={rgb, 255:red, 255; green, 255; blue, 255 }  ,fill opacity=1 ]   (422.04,76.43) .. controls (424.57,74.47) and (428.21,74.92) .. (430.18,77.44) .. controls (432.15,79.97) and (431.7,83.62) .. (429.17,85.58) .. controls (426.65,87.55) and (423,87.1) .. (421.03,84.57) .. controls (419.06,82.05) and (419.52,78.4) .. (422.04,76.43) -- cycle ;
			%Shape: Circle [id:dp8109023689731891] 
			\draw  [fill={rgb, 255:red, 255; green, 255; blue, 255 }  ,fill opacity=1 ]   (454.67,119.17) .. controls (456.23,121.96) and (455.24,125.5) .. (452.45,127.07) .. controls (449.65,128.64) and (446.12,127.64) .. (444.55,124.85) .. controls (442.98,122.06) and (443.98,118.52) .. (446.77,116.95) .. controls (449.56,115.38) and (453.1,116.38) .. (454.67,119.17) -- cycle ;
			%Shape: Circle [id:dp28076327346805563] 
			\draw  [fill={rgb, 255:red, 255; green, 255; blue, 255 }  ,fill opacity=1 ]   (508.26,139.65) .. controls (505.14,138.9) and (503.22,135.77) .. (503.97,132.66) .. controls (504.71,129.54) and (507.84,127.62) .. (510.96,128.37) .. controls (514.07,129.12) and (515.99,132.25) .. (515.25,135.36) .. controls (514.5,138.48) and (511.37,140.4) .. (508.26,139.65) -- cycle ;
			%Shape: Circle [id:dp3593171643654842] 
			\draw  [fill={rgb, 255:red, 255; green, 255; blue, 255 }  ,fill opacity=1 ]   (504.92,95.33) .. controls (501.99,96.61) and (498.57,95.26) .. (497.29,92.32) .. controls (496.01,89.39) and (497.36,85.97) .. (500.29,84.69) .. controls (503.23,83.41) and (506.65,84.76) .. (507.93,87.69) .. controls (509.2,90.63) and (507.86,94.05) .. (504.92,95.33) -- cycle ;
			%Shape: Circle [id:dp2609557460399563] 
			\draw  [fill={rgb, 255:red, 255; green, 255; blue, 255 }  ,fill opacity=1 ]   (460.81,53.01) .. controls (460.81,49.81) and (463.4,47.21) .. (466.61,47.21) .. controls (469.81,47.21) and (472.41,49.81) .. (472.41,53.01) .. controls (472.41,56.21) and (469.81,58.81) .. (466.61,58.81) .. controls (463.4,58.81) and (460.81,56.21) .. (460.81,53.01) -- cycle ;
			%Shape: Circle [id:dp5519771849591231] 
			\draw  [fill={rgb, 255:red, 255; green, 255; blue, 255 }  ,fill opacity=1 ]   (368.81,83.01) .. controls (368.81,79.81) and (371.4,77.21) .. (374.61,77.21) .. controls (377.81,77.21) and (380.41,79.81) .. (380.41,83.01) .. controls (380.41,86.21) and (377.81,88.81) .. (374.61,88.81) .. controls (371.4,88.81) and (368.81,86.21) .. (368.81,83.01) -- cycle ;
			%Shape: Circle [id:dp6939634223086303] 
			\draw  [fill={rgb, 255:red, 255; green, 255; blue, 255 }  ,fill opacity=1 ]   (476.41,178.58) .. controls (474.99,175.71) and (476.17,172.23) .. (479.04,170.81) .. controls (481.91,169.39) and (485.39,170.57) .. (486.81,173.44) .. controls (488.23,176.31) and (487.05,179.79) .. (484.18,181.21) .. controls (481.31,182.63) and (477.83,181.45) .. (476.41,178.58) -- cycle ;
			%Shape: Polygon Curved [id:ds49537003109265265] 
			\draw  [dash pattern={on 4.5pt off 4.5pt}] (268.3,82.6) .. controls (288.3,72.6) and (313.3,71.6) .. (339.3,84.6) .. controls (365.3,97.6) and (381.3,155.6) .. (379.3,180.6) .. controls (377.3,205.6) and (340.3,229.6) .. (314.3,233.6) .. controls (288.3,237.6) and (255.3,218.6) .. (237.3,201.6) .. controls (219.3,184.6) and (217.3,167.6) .. (222.3,137.6) .. controls (227.3,107.6) and (248.3,92.6) .. (268.3,82.6) -- cycle ;
			%Shape: Polygon Curved [id:ds5429176425616775] 
			\draw  [dash pattern={on 4.5pt off 4.5pt}] (422.3,38.6) .. controls (444.3,31.6) and (467.3,26.6) .. (493.3,39.6) .. controls (519.3,52.6) and (534.3,121.6) .. (532.3,146.6) .. controls (530.3,171.6) and (499.3,219.6) .. (453.3,189.6) .. controls (407.3,159.6) and (367.3,188.6) .. (332.3,164.6) .. controls (297.3,140.6) and (338.3,100.6) .. (359.3,77.6) .. controls (380.3,54.6) and (400.3,45.6) .. (422.3,38.6) -- cycle ;
			
			% Text Node
			\draw (324,137.4) node [anchor=north west][inner sep=0.75pt]    {$u_{1}$};
			% Text Node
			\draw (322,174.4) node [anchor=north west][inner sep=0.75pt]    {$u_{6}$};
			% Text Node
			\draw (289,186.4) node [anchor=north west][inner sep=0.75pt]    {$u_{5}$};
			% Text Node
			\draw (256,160.4) node [anchor=north west][inner sep=0.75pt]    {$u_{4}$};
			% Text Node
			\draw (256.81,129.17) node [anchor=north west][inner sep=0.75pt]    {$u_{3}$};
			% Text Node
			\draw (298,112.4) node [anchor=north west][inner sep=0.75pt]    {$u_{2}$};
			% Text Node
			\draw (190,129.4) node [anchor=north west][inner sep=0.75pt]    {$C_{0}$};
			% Text Node
			\draw (526,54.4) node [anchor=north west][inner sep=0.75pt]    {$T_{1}$};
			% Text Node
			\draw (90,110.4) node [anchor=north west][inner sep=0.75pt]    {\large $C$};

		\end{tikzpicture}
		
		\caption{A type-$2$ component~$C$. The dashed areas are the oriented cycle $C_0$ and the rooted tree~$T_1$ (rooted at~$u_1$). Here, we have~$t_1 = 9$, $t_2 = 2$, $t_3 = 4$, $t_4 = 3$, $t_5 = 1$ and $t_6 = 3$.}
		\label{fig:parity_type2}
	\end{figure}
	
	Let~$i \in \{1, \ldots, s\}$. Since no vertex rejects in~\ref{item:V1_parity}, for every vertex~$u$ of~$T_i$ which is a leaf, we have $n_u = 1$. By an immediate induction, for every vertex~$u$ in $T_i$ different from~$u_i$, $n_u$ is equal to the number of vertices in the subtree of~$T_i$ rooted at~$u$ modulo~$k$. Since $u_i$ has also~$u_{i-1 \mod s}$ as a child, we get~$n_{u_i} = t_i + n_{u_{i-1}} \mod k$ (if $i > 1$) and~$n_{u_1} = t_1 + n_{u_s} \mod k$. We thus obtain the following equalities.
	
	\begin{align*}
		n_{u_s} &= t_s + n_{u_{s-1}} \mod k\\
		&= t_s + t_{s-1} + n_{u_{s-2}} \mod k\\
		&= \ldots \\
		&= \sum_{i=2}^{s}t_{i} + n_{u_1} \mod k \\
		&= \sum_{i=1}^{s}t_{i} + n_{u_s} \mod k
	\end{align*}
	
	Thus, $\sum_{i=1}^{s}t_{i} = 0 \mod k$, so $|C|$ is divisible by~$k$.
	To conclude the proof, we proved that, for every connected component~$C$ of the underlying graph of~$\overrightarrow{G}$, $|C|$ is divisible by~$k$. Then~$|V(\overrightarrow{G})|$ is also divisible by~$k$. Finally, since~$G$ and~$\overrightarrow{G}$ have the same number of vertices, we get that~$|V(G)|$ is divisible by~$k$. So~$G$ satisfies~$\P$, which proves the soundness of the scheme. \qedhere

    %%%%%%%%%%%%%%%%%%%%%%%%%%%%%%%%%%%%
	\section{Upper bound in bounded degree and bounded expansion graphs}
    \label{sec:upper-bounded-expansion}
    %%%%%%%%%%%%%%%%%%%%%%%%%%%%%%%%%%%%
	
	In the previous section, we proved that for every~$k \geqslant 1$, the property~$\P$ of having a number of vertices divisible by~$k$ is a local property in the 2-local LCP model. Since this model is not the weakest possible, a natural question is whether the unique identifiers, and the verification radius equal to~$2$, are necessary or not for~$\P$ to be local.
	In Section~\ref{sec:lower-bound}, we partially answer this question, by showing that at least one of these two assumptions is necessary. More precisely, we prove that, in the $1$-local anonymous model, $\P$ is not local anymore (see Theorem~\ref{thm:parity lower bound}).
	
	However, an important remark is that, if there exists a spanning tree~$\T$ such that the prover can encode the parent of each vertex in the certificates, and if each vertex can recover its parent and its set of children thanks to the certificates, then, with only $O(\log k)$ additional bits (corresponding to~$n_u \mod k$ in the proof of Theorem~\ref{thm:upper-bound-general-graphs}), the vertices can verify the property~$\P$, by proceeding exactly as in the proof of Theorem~\ref{thm:upper-bound-general-graphs}.
	In the $2$-local LCP model, to do so, it was sufficient for the prover to write the distance to the root $d_u \mod 3$ in the certificate of each vertex~$u$ (then, the vertices used the identifiers and the verification radius~$2$ to recover their parent and children).
	In the~$1$-local anonymous model, even if the same strategy does not work anymore, we can wonder if it is possible to encode this parent relationship using certificates of constant size, in restricted graph classes.
	
	The first example of a graph class in which such a constant-size encoding is possible in the $1$-local anonymous model is in bounded-degree graphs. In graphs of maximum degree at most~$\Delta$, the prover can compute a proper distance-2 coloring (that is, a coloring of the vertices such that every vertices at distance at most~2 have different colors), using~$\Delta^2+1$ colors. Then, with~$O(\log \Delta)$ bits, it encodes, in the certificate of every vertex, its color. The vertices can check that this distance-2 coloring is indeed proper. Finally, for every vertex~$u$, all its neighbors have distinct colors, so the prover can encode its parent~$v$ simply by writing the color of~$v$ in the certificate of~$u$, using again $O(\log \Delta)$ bits. We formalize this in the following theorem.

	\begin{proposition}
		\label{prop:UB_parity_bounded_degree}
		Let~$k \geqslant 2$, and let $\P$ be the property of having a number of vertices divisible by~$k$. In graphs of maximum degree~$\Delta$, in the $1$-local anonymous model, the local complexity of~$\P$ is~$O(\log k + \log \Delta)$.
	\end{proposition}
	
	\begin{proof}
		%The proof of Proposition~\ref{prop:UB_parity_bounded_degree} is very similar to the proof of Theorem~\ref{thm:upper-bound-general-graphs}. The only difference is that, in the proof Theorem~\ref{thm:upper-bound-general_graphs}, the prover uses the identifiers to encode the parents in spanning tree~$\T$. Then, with verification radius~$2$, each vertex can determine its set of children. To prove Proposition~\ref{prop:UB_parity_bounded_degree}, the same approach is not possible anymore since we are in the $1$-local anonymous model. In bounded degree graphs, the strategy will simply be to for the prover to encode a proper distance-2 coloring in the certificates, and, for each vertex, to write the color of its parent in its certificate. This allows every vertex to know its parent and its set of children.
		
		%More formally,
		Let~$G$ be a graph of maximum degree~$\Delta$. Then, $G$ has a proper distance-2 coloring (that is, a coloring of the vertices such that every vertices at distance at most~$2$ have distinct colors) using at most~$\Delta^2+1$ colors. The prover chooses a spanning tree~$\T$ of~$G$ and a proper distance-2 coloring~$\varphi$ of~$G$. It writes, in the certificate of each vertex~$u$, the following information (which can be encoded on~$O(\log k + \log \Delta)$ bits):
		\begin{enumerate}[(C1)]
			\item the color~$\varphi(u)$ of~$u$; \label{item:C1_parity_bounded_degree}
			\item the color~$\varphi(v)$ of the parent~$v$ of~$u$ in~$\T$ (except if~$u$ is the root of~$\T$); \label{item:C2_parity_bounded_degree}
			\item the number of vertices~$n_u$ in the subtree of~$\T$ rooted at~$u$ modulo~$k$. \label{item:C3_parity_bounded_degree}
		\end{enumerate}
		The verification procedure of the vertices is %analogous as in the proof of Theorem~\ref{thm:upper_bound_general_graphs}. 
		the following. Let~$c$ be a certificate assignment. First, every vertex~$u$ checks that the colors written in~\ref{item:C1_parity_bounded_degree} in its certificate and in the certificates of its neighbors are all distinct. Then, we say that~$u$ is a~\emph{root} if the color of its parent is not written in~\ref{item:C2_parity_bounded_degree}. If~$u$ is not a root with respect to~$c$, it checks that it has a unique neighbor, its~\emph{parent}, colored with the color written in~\ref{item:C2_parity_bounded_degree}. The~\emph{children} of~$u$ are its neighbors~$v$ such that~$u$ is the parent of~$v$. Finally, $u$ is a \emph{leaf} if it does not have any children. Every vertex~$u$ can determine its set of children, which are its neighbors~$v$ such that the color of~$u$ is written in~\ref{item:C2_parity_bounded_degree} in the certificate of~$v$. The vertices verify that items~\ref{item:V1_parity} and~\ref{item:V2_parity} of the proof of Theorem~\ref{thm:upper-bound-general-graphs} both hold. All the rest of the proof (the completeness and soundness of the scheme) is identical as in Theorem~\ref{thm:upper-bound-general-graphs}. \qedhere
	\end{proof}

	We now generalize Proposition~\ref{prop:UB_parity_bounded_degree} to the larger class of graphs of \emph{bounded expansion}. Again, the idea is to find a spanning tree~$\T$ such that encoding the parent relationship in~$\T$ can be done with a constant number of bits.
	But first, let us define the class of graphs of bounded expansion.
	
	\begin{definition}
		Let~$G, G'$ be two graphs. We say that~$G'$ is a \emph{minor} of~$G$ if there exists a (non-necessarily induced) subgraph~$H$ of~$G$, and a partition of the vertices of~$H$ into connected sets, such that~$G'$ is the graph obtained from~$H$ by contracting these sets into single vertices. If each of these sets has diameter at most~$t$, we say that $G'$ is a \emph{$t$-shallow minor} of~$G$.
	\end{definition}

	\begin{definition}
		A class of graphs~$\C$ has \emph{bounded expansion} if there is a function $f$ such that, for each $t \geqslant 1$, each~$G \in \C$ and each $t$-shallow minor~$G'$ of~$G$, we have $|E(G')|/|V(G')| \leqslant f(t)$.
	\end{definition}

	Now, we aim to prove the following result:
	
	\begin{theorem}
		\label{thm:UB_parity_bounded_expansion}
		Let~$\C$ be a class of graphs of bounded expansion. Let~$k \geqslant 2$, and let~$\P$ be the property of having a number of vertices divisible by~$k$. In~$\C$, in the $1$-local anonymous model, the local complexity of~$\P$ is $O(\log k + g(\C))$, where $g(\C)$ is a constant that depends only on the class~$\C$.
	\end{theorem}
	
	The proof of Theorem~\ref{thm:UB_parity_bounded_expansion} relies on the notion of \emph{conflict-free coloring}, that we define now.
	
	\begin{definition}
		A \emph{conflict-free coloring} of a graph~$G$ is a coloring~$\varphi$ of~$V(G)$ such that, for every~$u \in V(G)$, there exists~$v \in N(u)$ such that $v$ is the only vertex in~$N(u)$ colored with~$\varphi(v)$. The \emph{conflict-free coloring number} of~$G$, denoted by~$\chi_{cf}(G)$, is the smallest integer~$c$ such that $G$ has a conflict-free coloring with $c$ colors.
	\end{definition}
	
	We will rely on the following result by Hickingbotham~\cite{Hickingbotham23}, which states that a class of bounded expansion has bounded conflict-free coloring number. 
	
	\begin{theorem}[\cite{Hickingbotham23}]
		\label{thm:bounded_expansion_conflict_free}
		Let $\C$ be a class of graphs of bounded expansion. Then, there exists a constant $f(\C)$ depending only on the class~$\C$, such that, for every~$G \in \C$, we have $\chi_{cf}(G) \leqslant f(\C)$.
	\end{theorem}
	
	\begin{remark}
		\label{rem:hereditary_conflict_free}
		If $\C$ is a class of graphs of bounded expansion, it is immediate from the definition that the class obtained from~$\C$ by taking all the subgraphs of graphs in~$\C$ has also bounded expansion. Thus, the result of Theorem~\ref{thm:bounded_expansion_conflict_free} remains true for all subgraphs of graphs in~$\C$.
	\end{remark}

	We are now able to prove Theorem~\ref{thm:UB_parity_bounded_expansion}.
	
	\begin{proof}[Proof of Theorem~\ref{thm:UB_parity_bounded_expansion}.]
		Let $G \in \C$.
		Let us first describe how the prover assigns the certificates on a correct instance. First, it chooses an arbitrary vertex \mbox{$r \in V(G)$}.
		Let~$\delta_r$ be the maximum distance from~$r$ to any vertex in~$G$, and, for every integer $1 \leqslant d \leqslant \delta_r$, let~$G_d$ be the bipartite graph obtained from~$G$ by keeping only the vertices at distance~$d-1$ and~$d$ from~$r$, and the edges between these two sets of vertices (we do not keep the edges between two vertices at distance~$d-1$, and between two vertices at distance~$d$). By Theorem~\ref{thm:bounded_expansion_conflict_free} and Remark~\ref{rem:hereditary_conflict_free}, $G_d$ has a conflict-free coloring using~$f(\C)$ colors, where $f(\C)$ depends only on~$\C$.
		The prover computes, for every $d \in \{1, \ldots, \delta_r\}$, such a conflict-free coloring~$\varphi_d$ of~$G_d$.
		Finally, the prover computes the following spanning tree~$\T$ of~$G$, rooted at~$r$. For every vertex~$u \neq r$ at distance~$d$ from~$r$, the parent of~$u$ in~$\T$ is an arbitrary neighbor~$v$ at distance~$d-1$ from~$r$, such that~$v$ is the only neighbor of~$u$ at distance~$d-1$ from~$r$ colored, in the coloring~$\varphi_d$, with the color~$\varphi_d(v)$. Since~$\varphi_d$ is a conflict-free coloring of~$G_d$, such a vertex~$v$ exists.
		Finally, for every $u \in V(G)$, the prover writes in the certificate of~$u$ the following information:
		
		\begin{enumerate}[(C1)]
			\item the distance~$d_u$ from~$u$ to~$r$ modulo~$3$; \label{item:C1_parity_bdexpansion}
			\item the number of vertices~$n_u$ in the subtree of~$\T$ rooted at~$u$ modulo~$k$; \label{item:C2_parity_bdexpansion}
			\item if $u \neq r$, the color $\varphi_{d_u}(v)$ of the parent~$v$ of~$u$ in~$\T$ in the coloring~$\varphi_{d_u}$; \label{item:C3_parity_bdexpansion}
			\item if $d_u < \delta_r$, the color $\varphi_{d_u+1}(u)$ of~$u$ in the coloring~$\varphi_{d_u+1}$. \label{item:C4_parity_bdexpansion}
		\end{enumerate}
		
		Given a certificate assignment~$c$, we say that a vertex~$u$ is a \emph{root} (with respect to the certificate assignment~$c$) if the set \mbox{$\{v \in N(u) \; | \; d_v = d_u-1 \mod 3\}$} is empty (where $d_u, d_v$ are the integers written in~\ref{item:C1_parity_bdexpansion} in the certificates of~$u$ and~$v$).
		Each vertex~$u$ performs the following verification procedure:
		\begin{enumerate}[(V1)]
			\item if $u$ is not a root, it checks that it has a unique neighbor~$v$ such that~$d_v = d_u-1 \mod 3$ and the color written in~\ref{item:C3_parity_bdexpansion} in the certificate of~$u$ is the same as the color written in~\ref{item:C4_parity_bdexpansion} in the certificate of~$v$.
			\label{item:V1_parity_bdexpansion}
		\end{enumerate}
		
		If no vertex rejects in~\ref{item:V1_parity_bdexpansion}, every non-root vertex~$u$ with respect to~$c$ has a unique neighbor~$v$ satisfying the condition in~\ref{item:V1_parity_bdexpansion}. This neighbor is called its~\emph{parent} (with respect to~$c$). If~$v$ is the parent of~$u$, we say that~$u$ is a~\emph{child} of~$v$. If~$u$ has no child, we say that it is a~\emph{leaf}. Note that each vertex~$u$ can determine its set of children. Indeed, the children of $u$ are its neighbors~$w$ such that the color of~$w$ written in~\ref{item:C3_parity_bdexpansion} in the certificate of~$w$ is the same as the color of~$u$ written in~\ref{item:C4_parity_bdexpansion} in the certificate of~$u$. The verification of~$u$ then continues with the two following steps:
		\begin{enumerate}[(V1)]
			\setcounter{enumi}{1}
			\item if $u$ is a root, it checks that $n_u = 0 \mod k$; \label{item:V2_parity_bdexpansion}
			\item in all cases, $u$ checks that~$n_u = \sum_{v \in S} n_v \mod k$, where $S$ denotes its set of children. \label{item:V3_parity_bdexpansion}
		\end{enumerate}
		
		The proof of the correctness of the scheme is identical as in the proof of Theorem~\ref{thm:upper-bound-general-graphs}. \qedhere
	\end{proof}
	
\end{proof}

%%%%%%%%%%%%%%%%%%%%%%%%%%%%%%%%%
\section{Lower bound}
\label{sec:lower-bound}
%%%%%%%%%%%%%%%%%%%%%%%%%%%%%%%%%

We now move on to our lower bound. 

\begin{theorem}
	\label{thm:parity lower bound}
	In the $1$-local anonymous model, the local complexity of having an even number of vertices is $\omega(1)$.
\end{theorem}

To prove Theorem~\ref{thm:parity lower bound}, we will need the following result from Ramsey theory, called the Folkman-Rado-Sanders Finite Unions Theorem. The classic proof can be found in~\cite[Chapter~3]{Graham91}.

\begin{theorem}[Finite Unions Theorem]
	\label{thm:finite unions}
	For every integers $k, c$, there exists an integer $m = f(k, c)$ such that, if all the non-empty subsets of~$[m]$ are colored with $c$ colors, then there exists~$k$ pairwise disjoint non-empty subsets~$S_1, \ldots, S_k \subseteq [n]$ such that \mbox{$\{\cup_{i \in I} S_i \; | \; \emptyset \subsetneq I \subseteq [k]\}$} (in other words, all non-empty unions of these sets) are colored with the same color.
\end{theorem}

We are in particular interested with the case $k=2$ in Theorem~\ref{thm:finite unions}, which guarantees the existence of two disjoint non-empty subsets $A, B \subseteq [m]$ such that $A$, $B$ and $A \cup B$ have the same color, where~$m := f(2,c)$. Using this result, we are able to prove Theorem~\ref{thm:parity lower bound}

\begin{proof}[Proof of Theorem~\ref{thm:parity lower bound}.]
	By contradiction, assume that it is possible to certify that a graph has an even number of vertices, in the 1-local anonymous model, with certificates of constant size~$s$. Let $c: = 2^s$ be the number of possible different certificates in such a certification scheme.
	Let $m \geqslant f(2,c)$ be an odd integer, where~$f$ is the function provided by Theorem~\ref{thm:finite unions}, and finally let $n := m + 2^m -1$. Note that $n$ is even. Consider the $n$-vertex bipartite graph~$G_m$, where the vertices of $G_m$ are partitioned into two sets $X,Y$ (to simplify the notations, we do not index~$X$ and~$Y$ by~$m$). Let $X := [m]$ and $Y$ be the set of non-empty subsets of $[m]$. For every $x \in X$ and $y \in Y$, there is an edge between $x$ and $y$ in~$G_m$ if and only if $x \in y$.
	
	Since $G_m$ has an even number of vertices, there exists a certificate assignment to its vertices such that each vertex accepts. Consider such a certificate assignment, with certificates in~$[c]$. Consider the following coloring of the non-empty subsets of~$[m]$. Let $\emptyset \subseteq S \subseteq [m]$. $S$ is also a non-empty subset of vertices of~$X$. By construction of $G_m$, there is a unique vertex~$y \in Y$ such the neighbors of~$y$ in~$X$ are exactly the vertices in~$S$. We define the color of~$S$ as being the certificate of~$y$.
	
	By Theorem~\ref{thm:finite unions}, there exist two disjoint non-empty subsets $A, B \subseteq [m]$ such that $A$, $B$, and $A \cup B$ all have the same color. In other words, there exist three vertices $y_A, y_B, y_{A \cup B} \in Y$ such that:
	\begin{itemize}
		\item $y_A$, $y_B$, $y_{A \cup B}$ all have the same certificate, and
		\item the neighbors of $y_A$ (resp of $y_B$, $y_{A \cup B}$) in~$X$ are the vertices in~$A$ (resp in $B$, $A \cup B$).
	\end{itemize}
	
	We construct the following graph $G_m'$, together with a certificate assignment such that every vertex accepts. We take two copies of~$G_m$, each one equipped with the accepting certificate assignment of~$G_m$. Then, we identify the two vertices corresponding to $y_A$ in the first copy, and to $y_B$ in the second copy. The graph resulting from this operation is~$G_m'$. Note that this operation is compatible with the certificate assignment because $y_A$ and $y_B$ have the same certificates. Moreover, $G_m'$ has $2n-1$ vertices, so it has an odd number of vertices.
	Finally, every vertex accepts in $G_m'$, because the vertex obtained by identifying the vertices $y_A$ and $y_B$ in the two copies has the same view as $y_{A \cup B}$ in $G_m$, and all the other vertices have the same view as they already have in~$G_m$. This is a contradiction and concludes the proof. See Figure~\ref{fig:lower_anonymous} for an illustration. \qedhere

	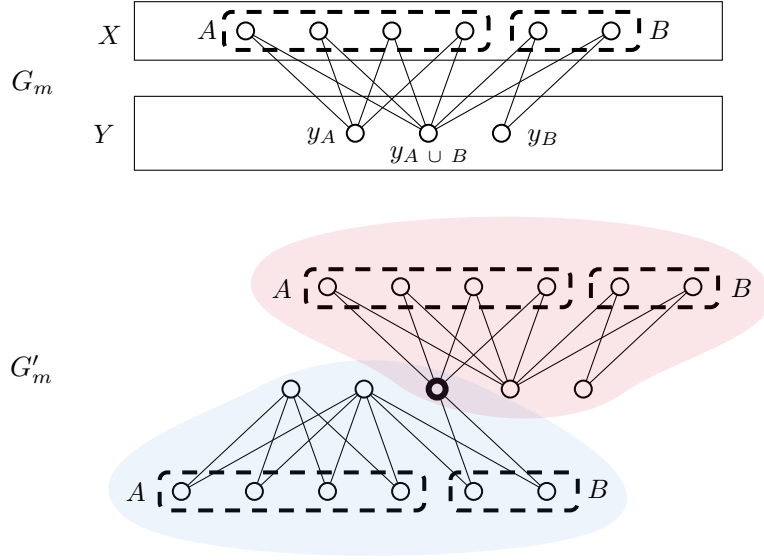
\begin{figure}
		\centering
		
		\begin{tikzpicture}[x=0.55pt,y=0.55pt,yscale=-1,xscale=1]
			%uncomment if require: \path (0,456); %set diagram left start at 0, and has height of 456
			
			%Straight Lines [id:da6059746948914558] 
			\draw    (350.8,302.8) -- (425.8,372.8) ;
			%Straight Lines [id:da9976063106434676] 
			\draw    (350.8,302.8) -- (375.8,372.8) ;
			%Straight Lines [id:da3849731558487043] 
			\draw    (300.8,302.8) -- (425.8,372.8) ;
			%Straight Lines [id:da37060455737411124] 
			\draw    (300.8,302.8) -- (375.8,372.8) ;
			%Straight Lines [id:da8191344736487329] 
			\draw    (300.8,302.8) -- (325.8,372.8) ;
			%Straight Lines [id:da47624850835220056] 
			\draw    (300.8,302.8) -- (275.8,372.8) ;
			%Straight Lines [id:da3075740797237452] 
			\draw    (300.8,302.8) -- (225.8,372.8) ;
			%Straight Lines [id:da786777517221951] 
			\draw    (300.8,302.8) -- (175.8,372.8) ;
			%Straight Lines [id:da06032391894609013] 
			\draw    (250.8,302.8) -- (325.8,372.8) ;
			%Straight Lines [id:da1091636707002942] 
			\draw    (250.8,302.8) -- (275.8,372.8) ;
			%Straight Lines [id:da9621357056942232] 
			\draw    (250.8,302.8) -- (225.8,372.8) ;
			%Straight Lines [id:da41267868607665015] 
			\draw    (250.8,302.8) -- (175.8,372.8) ;
			%Straight Lines [id:da9201210581688456] 
			\draw    (219.8,57.8) -- (294.8,127.8) ;
			%Straight Lines [id:da4547148876133331] 
			\draw    (319.8,57.8) -- (294.8,127.8) ;
			%Straight Lines [id:da5335408216794831] 
			\draw    (269.8,57.8) -- (294.8,127.8) ;
			%Straight Lines [id:da5340157921957689] 
			\draw    (369.8,57.8) -- (294.8,127.8) ;
			%Straight Lines [id:da0004039019550761136] 
			\draw    (419.8,57.8) -- (394.8,127.8) ;
			%Straight Lines [id:da07317705788847373] 
			\draw    (469.8,57.8) -- (394.8,127.8) ;
			%Straight Lines [id:da01446694389808001] 
			\draw    (219.8,57.8) -- (344.8,127.8) ;
			%Straight Lines [id:da17057756168763494] 
			\draw    (269.8,57.8) -- (344.8,127.8) ;
			%Straight Lines [id:da4164771508524785] 
			\draw    (319.8,57.8) -- (344.8,127.8) ;
			%Straight Lines [id:da9802056315418446] 
			\draw    (369.8,57.8) -- (344.8,127.8) ;
			%Straight Lines [id:da948179682400575] 
			\draw    (419.8,57.8) -- (344.8,127.8) ;
			%Straight Lines [id:da37012029420771675] 
			\draw    (469.8,57.8) -- (344.8,127.8) ;
			%Shape: Circle [id:dp37930508300930144] 
			\draw  [fill={rgb, 255:red, 255; green, 255; blue, 255 }  ,fill opacity=1 ][line width=0.75]  (289,127.8) .. controls (289,124.6) and (291.6,122) .. (294.8,122) .. controls (298,122) and (300.6,124.6) .. (300.6,127.8) .. controls (300.6,131) and (298,133.6) .. (294.8,133.6) .. controls (291.6,133.6) and (289,131) .. (289,127.8) -- cycle ;
			%Shape: Circle [id:dp2629082048056386] 
			\draw  [fill={rgb, 255:red, 255; green, 255; blue, 255 }  ,fill opacity=1 ][line width=0.75]  (389,127.8) .. controls (389,124.6) and (391.6,122) .. (394.8,122) .. controls (398,122) and (400.6,124.6) .. (400.6,127.8) .. controls (400.6,131) and (398,133.6) .. (394.8,133.6) .. controls (391.6,133.6) and (389,131) .. (389,127.8) -- cycle ;
			%Shape: Circle [id:dp8201172493322982] 
			\draw  [fill={rgb, 255:red, 255; green, 255; blue, 255 }  ,fill opacity=1 ][line width=0.75]  (339,127.8) .. controls (339,124.6) and (341.6,122) .. (344.8,122) .. controls (348,122) and (350.6,124.6) .. (350.6,127.8) .. controls (350.6,131) and (348,133.6) .. (344.8,133.6) .. controls (341.6,133.6) and (339,131) .. (339,127.8) -- cycle ;
			%Shape: Circle [id:dp35245455667334225] 
			\draw  [fill={rgb, 255:red, 255; green, 255; blue, 255 }  ,fill opacity=1 ][line width=0.75]  (264,57.8) .. controls (264,54.6) and (266.6,52) .. (269.8,52) .. controls (273,52) and (275.6,54.6) .. (275.6,57.8) .. controls (275.6,61) and (273,63.6) .. (269.8,63.6) .. controls (266.6,63.6) and (264,61) .. (264,57.8) -- cycle ;
			%Shape: Circle [id:dp21780470777018945] 
			\draw  [fill={rgb, 255:red, 255; green, 255; blue, 255 }  ,fill opacity=1 ][line width=0.75]  (314,57.8) .. controls (314,54.6) and (316.6,52) .. (319.8,52) .. controls (323,52) and (325.6,54.6) .. (325.6,57.8) .. controls (325.6,61) and (323,63.6) .. (319.8,63.6) .. controls (316.6,63.6) and (314,61) .. (314,57.8) -- cycle ;
			%Shape: Circle [id:dp18783010793301158] 
			\draw  [fill={rgb, 255:red, 255; green, 255; blue, 255 }  ,fill opacity=1 ][line width=0.75]  (214,57.8) .. controls (214,54.6) and (216.6,52) .. (219.8,52) .. controls (223,52) and (225.6,54.6) .. (225.6,57.8) .. controls (225.6,61) and (223,63.6) .. (219.8,63.6) .. controls (216.6,63.6) and (214,61) .. (214,57.8) -- cycle ;
			%Shape: Circle [id:dp58733842927079] 
			\draw  [fill={rgb, 255:red, 255; green, 255; blue, 255 }  ,fill opacity=1 ][line width=0.75]  (364,57.8) .. controls (364,54.6) and (366.6,52) .. (369.8,52) .. controls (373,52) and (375.6,54.6) .. (375.6,57.8) .. controls (375.6,61) and (373,63.6) .. (369.8,63.6) .. controls (366.6,63.6) and (364,61) .. (364,57.8) -- cycle ;
			%Shape: Circle [id:dp3175630626313398] 
			\draw  [fill={rgb, 255:red, 255; green, 255; blue, 255 }  ,fill opacity=1 ][line width=0.75]  (414,57.8) .. controls (414,54.6) and (416.6,52) .. (419.8,52) .. controls (423,52) and (425.6,54.6) .. (425.6,57.8) .. controls (425.6,61) and (423,63.6) .. (419.8,63.6) .. controls (416.6,63.6) and (414,61) .. (414,57.8) -- cycle ;
			%Shape: Circle [id:dp6480423946228673] 
			\draw  [fill={rgb, 255:red, 255; green, 255; blue, 255 }  ,fill opacity=1 ][line width=0.75]  (464,57.8) .. controls (464,54.6) and (466.6,52) .. (469.8,52) .. controls (473,52) and (475.6,54.6) .. (475.6,57.8) .. controls (475.6,61) and (473,63.6) .. (469.8,63.6) .. controls (466.6,63.6) and (464,61) .. (464,57.8) -- cycle ;
			%Shape: Rectangle [id:dp012946642496580285] 
			\draw   (143.7,102.53) -- (545.9,102.53) -- (545.9,153.08) -- (143.7,153.08) -- cycle ;
			%Shape: Rectangle [id:dp2794977369300029] 
			\draw   (143.7,37.8) -- (545.9,37.8) -- (545.9,77.8) -- (143.7,77.8) -- cycle ;
			%Rounded Rect [id:dp16978401143326316] 
			\draw  [dash pattern={on 5.63pt off 4.5pt}][line width=1.5]  (205,50) .. controls (205,47.13) and (207.33,44.8) .. (210.2,44.8) -- (380.8,44.8) .. controls (383.67,44.8) and (386,47.13) .. (386,50) -- (386,65.6) .. controls (386,68.47) and (383.67,70.8) .. (380.8,70.8) -- (210.2,70.8) .. controls (207.33,70.8) and (205,68.47) .. (205,65.6) -- cycle ;
			%Rounded Rect [id:dp7647068091676057] 
			\draw  [dash pattern={on 5.63pt off 4.5pt}][line width=1.5]  (402,50) .. controls (402,47.13) and (404.33,44.8) .. (407.2,44.8) -- (483.8,44.8) .. controls (486.67,44.8) and (489,47.13) .. (489,50) -- (489,65.6) .. controls (489,68.47) and (486.67,70.8) .. (483.8,70.8) -- (407.2,70.8) .. controls (404.33,70.8) and (402,68.47) .. (402,65.6) -- cycle ;
			%Straight Lines [id:da20427749801626371] 
			\draw    (275.8,232.8) -- (350.8,302.8) ;
			%Straight Lines [id:da23053866166185277] 
			\draw    (375.8,232.8) -- (350.8,302.8) ;
			%Straight Lines [id:da36177726511767583] 
			\draw    (325.8,232.8) -- (350.8,302.8) ;
			%Straight Lines [id:da8425873006853286] 
			\draw    (425.8,232.8) -- (350.8,302.8) ;
			%Straight Lines [id:da3897228791144576] 
			\draw    (475.8,232.8) -- (450.8,302.8) ;
			%Straight Lines [id:da22408473737704926] 
			\draw    (525.8,232.8) -- (450.8,302.8) ;
			%Straight Lines [id:da3517532711459589] 
			\draw    (275.8,232.8) -- (400.8,302.8) ;
			%Straight Lines [id:da07479706845401823] 
			\draw    (325.8,232.8) -- (400.8,302.8) ;
			%Straight Lines [id:da01753115787804882] 
			\draw    (375.8,232.8) -- (400.8,302.8) ;
			%Straight Lines [id:da7173387804314355] 
			\draw    (425.8,232.8) -- (400.8,302.8) ;
			%Straight Lines [id:da6840108418388868] 
			\draw    (475.8,232.8) -- (400.8,302.8) ;
			%Straight Lines [id:da37429117830014336] 
			\draw    (525.8,232.8) -- (400.8,302.8) ;
			%Shape: Circle [id:dp6049627651102505] 
			\draw  [fill={rgb, 255:red, 255; green, 255; blue, 255 }  ,fill opacity=1 ][line width=2.25]  (345,302.8) .. controls (345,299.6) and (347.6,297) .. (350.8,297) .. controls (354,297) and (356.6,299.6) .. (356.6,302.8) .. controls (356.6,306) and (354,308.6) .. (350.8,308.6) .. controls (347.6,308.6) and (345,306) .. (345,302.8) -- cycle ;
			%Shape: Circle [id:dp1977812280368284] 
			\draw  [fill={rgb, 255:red, 255; green, 255; blue, 255 }  ,fill opacity=1 ][line width=0.75]  (445,302.8) .. controls (445,299.6) and (447.6,297) .. (450.8,297) .. controls (454,297) and (456.6,299.6) .. (456.6,302.8) .. controls (456.6,306) and (454,308.6) .. (450.8,308.6) .. controls (447.6,308.6) and (445,306) .. (445,302.8) -- cycle ;
			%Shape: Circle [id:dp13288905123030748] 
			\draw  [fill={rgb, 255:red, 255; green, 255; blue, 255 }  ,fill opacity=1 ][line width=0.75]  (395,302.8) .. controls (395,299.6) and (397.6,297) .. (400.8,297) .. controls (404,297) and (406.6,299.6) .. (406.6,302.8) .. controls (406.6,306) and (404,308.6) .. (400.8,308.6) .. controls (397.6,308.6) and (395,306) .. (395,302.8) -- cycle ;
			%Shape: Circle [id:dp9916268455073324] 
			\draw  [fill={rgb, 255:red, 255; green, 255; blue, 255 }  ,fill opacity=1 ][line width=0.75]  (320,232.8) .. controls (320,229.6) and (322.6,227) .. (325.8,227) .. controls (329,227) and (331.6,229.6) .. (331.6,232.8) .. controls (331.6,236) and (329,238.6) .. (325.8,238.6) .. controls (322.6,238.6) and (320,236) .. (320,232.8) -- cycle ;
			%Shape: Circle [id:dp7008331789360446] 
			\draw  [fill={rgb, 255:red, 255; green, 255; blue, 255 }  ,fill opacity=1 ][line width=0.75]  (370,232.8) .. controls (370,229.6) and (372.6,227) .. (375.8,227) .. controls (379,227) and (381.6,229.6) .. (381.6,232.8) .. controls (381.6,236) and (379,238.6) .. (375.8,238.6) .. controls (372.6,238.6) and (370,236) .. (370,232.8) -- cycle ;
			%Shape: Circle [id:dp35159848807877103] 
			\draw  [fill={rgb, 255:red, 255; green, 255; blue, 255 }  ,fill opacity=1 ][line width=0.75]  (270,232.8) .. controls (270,229.6) and (272.6,227) .. (275.8,227) .. controls (279,227) and (281.6,229.6) .. (281.6,232.8) .. controls (281.6,236) and (279,238.6) .. (275.8,238.6) .. controls (272.6,238.6) and (270,236) .. (270,232.8) -- cycle ;
			%Shape: Circle [id:dp8382471414424687] 
			\draw  [fill={rgb, 255:red, 255; green, 255; blue, 255 }  ,fill opacity=1 ][line width=0.75]  (420,232.8) .. controls (420,229.6) and (422.6,227) .. (425.8,227) .. controls (429,227) and (431.6,229.6) .. (431.6,232.8) .. controls (431.6,236) and (429,238.6) .. (425.8,238.6) .. controls (422.6,238.6) and (420,236) .. (420,232.8) -- cycle ;
			%Shape: Circle [id:dp3125647161030134] 
			\draw  [fill={rgb, 255:red, 255; green, 255; blue, 255 }  ,fill opacity=1 ][line width=0.75]  (470,232.8) .. controls (470,229.6) and (472.6,227) .. (475.8,227) .. controls (479,227) and (481.6,229.6) .. (481.6,232.8) .. controls (481.6,236) and (479,238.6) .. (475.8,238.6) .. controls (472.6,238.6) and (470,236) .. (470,232.8) -- cycle ;
			%Shape: Circle [id:dp6880336611986688] 
			\draw  [fill={rgb, 255:red, 255; green, 255; blue, 255 }  ,fill opacity=1 ][line width=0.75]  (520,232.8) .. controls (520,229.6) and (522.6,227) .. (525.8,227) .. controls (529,227) and (531.6,229.6) .. (531.6,232.8) .. controls (531.6,236) and (529,238.6) .. (525.8,238.6) .. controls (522.6,238.6) and (520,236) .. (520,232.8) -- cycle ;
			%Shape: Circle [id:dp4269688634311467] 
			\draw  [fill={rgb, 255:red, 255; green, 255; blue, 255 }  ,fill opacity=1 ][line width=0.75]  (295,302.8) .. controls (295,299.6) and (297.6,297) .. (300.8,297) .. controls (304,297) and (306.6,299.6) .. (306.6,302.8) .. controls (306.6,306) and (304,308.6) .. (300.8,308.6) .. controls (297.6,308.6) and (295,306) .. (295,302.8) -- cycle ;
			%Shape: Circle [id:dp6788705630350758] 
			\draw  [fill={rgb, 255:red, 255; green, 255; blue, 255 }  ,fill opacity=1 ][line width=0.75]  (245,302.8) .. controls (245,299.6) and (247.6,297) .. (250.8,297) .. controls (254,297) and (256.6,299.6) .. (256.6,302.8) .. controls (256.6,306) and (254,308.6) .. (250.8,308.6) .. controls (247.6,308.6) and (245,306) .. (245,302.8) -- cycle ;
			%Shape: Circle [id:dp5873505985382534] 
			\draw  [fill={rgb, 255:red, 255; green, 255; blue, 255 }  ,fill opacity=1 ][line width=0.75]  (170,372.8) .. controls (170,369.6) and (172.6,367) .. (175.8,367) .. controls (179,367) and (181.6,369.6) .. (181.6,372.8) .. controls (181.6,376) and (179,378.6) .. (175.8,378.6) .. controls (172.6,378.6) and (170,376) .. (170,372.8) -- cycle ;
			%Shape: Circle [id:dp15592296630200808] 
			\draw  [fill={rgb, 255:red, 255; green, 255; blue, 255 }  ,fill opacity=1 ][line width=0.75]  (220,372.8) .. controls (220,369.6) and (222.6,367) .. (225.8,367) .. controls (229,367) and (231.6,369.6) .. (231.6,372.8) .. controls (231.6,376) and (229,378.6) .. (225.8,378.6) .. controls (222.6,378.6) and (220,376) .. (220,372.8) -- cycle ;
			%Shape: Circle [id:dp9929188658039186] 
			\draw  [fill={rgb, 255:red, 255; green, 255; blue, 255 }  ,fill opacity=1 ][line width=0.75]  (270,372.8) .. controls (270,369.6) and (272.6,367) .. (275.8,367) .. controls (279,367) and (281.6,369.6) .. (281.6,372.8) .. controls (281.6,376) and (279,378.6) .. (275.8,378.6) .. controls (272.6,378.6) and (270,376) .. (270,372.8) -- cycle ;
			%Shape: Circle [id:dp810960489976672] 
			\draw  [fill={rgb, 255:red, 255; green, 255; blue, 255 }  ,fill opacity=1 ][line width=0.75]  (320,372.8) .. controls (320,369.6) and (322.6,367) .. (325.8,367) .. controls (329,367) and (331.6,369.6) .. (331.6,372.8) .. controls (331.6,376) and (329,378.6) .. (325.8,378.6) .. controls (322.6,378.6) and (320,376) .. (320,372.8) -- cycle ;
			%Shape: Circle [id:dp4694138441791157] 
			\draw  [fill={rgb, 255:red, 255; green, 255; blue, 255 }  ,fill opacity=1 ][line width=0.75]  (370,372.8) .. controls (370,369.6) and (372.6,367) .. (375.8,367) .. controls (379,367) and (381.6,369.6) .. (381.6,372.8) .. controls (381.6,376) and (379,378.6) .. (375.8,378.6) .. controls (372.6,378.6) and (370,376) .. (370,372.8) -- cycle ;
			%Shape: Circle [id:dp05290624230700369] 
			\draw  [fill={rgb, 255:red, 255; green, 255; blue, 255 }  ,fill opacity=1 ][line width=0.75]  (420,372.8) .. controls (420,369.6) and (422.6,367) .. (425.8,367) .. controls (429,367) and (431.6,369.6) .. (431.6,372.8) .. controls (431.6,376) and (429,378.6) .. (425.8,378.6) .. controls (422.6,378.6) and (420,376) .. (420,372.8) -- cycle ;
			%Rounded Rect [id:dp1379047173131367] 
			\draw  [dash pattern={on 5.63pt off 4.5pt}][line width=1.5]  (261,226) .. controls (261,223.13) and (263.33,220.8) .. (266.2,220.8) -- (436.8,220.8) .. controls (439.67,220.8) and (442,223.13) .. (442,226) -- (442,241.6) .. controls (442,244.47) and (439.67,246.8) .. (436.8,246.8) -- (266.2,246.8) .. controls (263.33,246.8) and (261,244.47) .. (261,241.6) -- cycle ;
			%Rounded Rect [id:dp8824235094365718] 
			\draw  [dash pattern={on 5.63pt off 4.5pt}][line width=1.5]  (456,226) .. controls (456,223.13) and (458.33,220.8) .. (461.2,220.8) -- (537.8,220.8) .. controls (540.67,220.8) and (543,223.13) .. (543,226) -- (543,241.6) .. controls (543,244.47) and (540.67,246.8) .. (537.8,246.8) -- (461.2,246.8) .. controls (458.33,246.8) and (456,244.47) .. (456,241.6) -- cycle ;
			%Rounded Rect [id:dp10528637301168964] 
			\draw  [dash pattern={on 5.63pt off 4.5pt}][line width=1.5]  (161,365) .. controls (161,362.13) and (163.33,359.8) .. (166.2,359.8) -- (336.8,359.8) .. controls (339.67,359.8) and (342,362.13) .. (342,365) -- (342,380.6) .. controls (342,383.47) and (339.67,385.8) .. (336.8,385.8) -- (166.2,385.8) .. controls (163.33,385.8) and (161,383.47) .. (161,380.6) -- cycle ;
			%Rounded Rect [id:dp2567588766500404] 
			\draw  [dash pattern={on 5.63pt off 4.5pt}][line width=1.5]  (360,365) .. controls (360,362.13) and (362.33,359.8) .. (365.2,359.8) -- (441.8,359.8) .. controls (444.67,359.8) and (447,362.13) .. (447,365) -- (447,380.6) .. controls (447,383.47) and (444.67,385.8) .. (441.8,385.8) -- (365.2,385.8) .. controls (362.33,385.8) and (360,383.47) .. (360,380.6) -- cycle ;
			%Shape: Polygon Curved [id:ds7046206704528585] 
			\draw  [draw opacity=0][fill={rgb, 255:red, 208; green, 2; blue, 27 }  ,fill opacity=0.1 ] (231,214.7) .. controls (266,177.7) and (510,171.7) .. (564,215.7) .. controls (618,259.7) and (519,276.7) .. (478,305.7) .. controls (437,334.7) and (354,323.7) .. (322,298.7) .. controls (290,273.7) and (196,251.7) .. (231,214.7) -- cycle ;
			%Shape: Polygon Curved [id:ds2655304463102248] 
			\draw  [draw opacity=0][fill={rgb, 255:red, 74; green, 144; blue, 226 }  ,fill opacity=0.1 ] (228,301.7) .. controls (263,282.7) and (317,273.7) .. (371,296.7) .. controls (425,319.7) and (511,365.7) .. (470,394.7) .. controls (429,423.7) and (174,426.7) .. (135,388.7) .. controls (96,350.7) and (193,320.7) .. (228,301.7) -- cycle ;
			
			% Text Node
			\draw (117,52.4) node [anchor=north west][inner sep=0.75pt]    {$X$};
			% Text Node
			\draw (114,120.4) node [anchor=north west][inner sep=0.75pt]    {$Y$};
			% Text Node
			\draw (185,49.4) node [anchor=north west][inner sep=0.75pt]    {$A$};
			% Text Node
			\draw (494,49.4) node [anchor=north west][inner sep=0.75pt]    {$B$};
			% Text Node
			\draw (260,122.4) node [anchor=north west][inner sep=0.75pt]    {$y_{A}$};
			% Text Node
			\draw (411,123.4) node [anchor=north west][inner sep=0.75pt]    {$y_{B}$};
			% Text Node
			\draw (316.8,135.2) node [anchor=north west][inner sep=0.75pt]    {$y_{A\ \cup \ B}$};
			% Text Node
			\draw (236,226.4) node [anchor=north west][inner sep=0.75pt]    {$A$};
			% Text Node
			\draw (550,225.4) node [anchor=north west][inner sep=0.75pt]    {$B$};
			% Text Node
			\draw (136,365.4) node [anchor=north west][inner sep=0.75pt]    {$A$};
			% Text Node
			\draw (452,363.4) node [anchor=north west][inner sep=0.75pt]    {$B$};
			% Text Node
			\draw (58,83.4) node [anchor=north west][inner sep=0.75pt]    { \large $G_m$};
			% Text Node
			\draw (57,277.7) node [anchor=north west][inner sep=0.75pt]    { \large$G_m'$};

		\end{tikzpicture}
		
		\caption{The graphs $G_m$ and $G_m'$. In $G_m$, the only vertices of $X$ which are represented are vertices in~$A \cup B$, and the only vertices of~$Y$ which are represented are $y_A$, $y_B$, $y_{A \cup B}$ (and similarly in $G_m'$). The vertices $y_A$, $y_B$, $y_{A \cup B}$ have the same certificates. The graph $G_m'$ is obtained by taking two copies of~$G_m$ (each one represented in one of the colored areas), and identifying $y_A$ in the first copy with $y_B$ in the second one. The corresponding vertex is the one drawn in bold. Each vertex accepts in~$G_m'$ because it has the same view as a vertex which accepts in~$G_m$.}
		\label{fig:lower_anonymous}
	\end{figure}
	
\end{proof}

	Now, one can wonder if it is possible to obtain a more precise lower bound in Theorem~\ref{thm:parity lower bound}.
	To do so, with this proof technique, we would need to obtain an upper bound on $f(2,c)$ in Theorem~\ref{thm:finite unions}.
	Such an upper bound is proven in~\cite[Chapter 6]{Setyawan98}. More precisely, a general upper bound on $f(k,c)$ is proven, and the particular case for $k=2$ is the following.
	
	\begin{theorem}[\cite{Setyawan98}]
		\label{thm:bound_finite_unions}
		The following inequality holds:
		$$f(2,c) \leqslant c^{3^{c^{\iddots^3}}}$$
		where $c^{3^{c^{\iddots^3}}}$ is a tower of exponential of height $2c$.
	\end{theorem}

	Using Theorem~\ref{thm:bound_finite_unions}, we are able to get a more precise lower bound on the local complexity of having an even number of vertices.
	For every integers $k, \ell \geqslant 1$, let $\tow(\ell, k) := \ell^{\ell^{\iddots^\ell}}$, where the exponential tower has size~$k$ (for instance, \mbox{$\tow(\ell, 3) = \ell^{\ell^{\ell}}$}).
	Recall that for every integer~$n \geqslant 2$, the \emph{iterated logarithm of~$n$}, denoted by $\log^\ast(n)$, is the smallest integer~$k\geqslant 1$ such that $\tow(2,k) \geqslant n$. In other words, it is the smallest number of times we have to apply the (base-2) logarithm to the integer~$n$ to get a value below~$1$.
	We will prove the following lower bound.

	\begin{theorem}
		\label{thm:explicit_lower_bound_parity}
		In the $1$-local anonymous model, for $n$-vertex graphs, at least $\Omega(\log^\ast(n))$ different certificates are necessary to certify that the number of vertices is even. Thus, in this model, the local complexity of this property is $\Omega(\log(\log^\ast(n)))$.
	\end{theorem}

	To prove Theorem~\ref{thm:explicit_lower_bound_parity}, we will use the bound of Theorem~\ref{thm:bound_finite_unions}. Since the tower of exponential appearing in this bound does not have a constant exponent, we will need to prove first a technical result. Let denote by $\dlog(n)$ (the \emph{diagonal iterated logarithm of~$n$}) the smallest integer~$k\geqslant 1$ such that~$\tow(k,k) \geqslant n$.
	Before proving Theorem~\ref{thm:explicit_lower_bound_parity}, let us prove the following lemma, that gives a comparison between~$\log^\ast$ and~$\dlog$.
	
	\begin{lemma}
		\label{lem:log*}
		We have $\dlog(n) = \Theta(\log^\ast(n))$.
	\end{lemma}
	
	\begin{proof}
		Since we trivially have $\dlog(n) = O(\log^\ast(n))$ from the definition, we just need to show that $\log^\ast(n) = O(\dlog(n))$. Let $m \geqslant 4$. For every~$k \in \N$, let~$\log^{(k)}(m) := \log \log \ldots \log m$, where the logarithm is in base~$2$, and is iterated~$k$ times (of course, if it exists). Note that, for every integers $1 \leqslant k \leqslant m$, $\log^{(k)}(\tow(m,m))$ exists. By induction on~$k \in [m-1]$, let us prove the following inequality: $$\log^{(k)}(\tow(m,m)) \leqslant 2 \tow(m,m-k) \log m $$
		
		The result is straightforward for~$k = 1$, since \mbox{$\log \tow(m,m) = \tow(m,m-1) \log m$}. Let $k \in \{2, \ldots, m-1\}$, and assume that the result holds for $k-1$. In other words, assume that:
		$$\log^{(k-1)}(\tow(m,m)) \leqslant 2 \tow(m,m-k+1) \log m $$
		
		By taking the logarithm of this expression, we get:
		\begin{align*}
			\log^{(k)}(\tow(m,m))
			& \leqslant 1 + \log \log m + \tow(m,m-k) \log m \\
			& \leqslant 2 \tow(m,m-k) \log m \\
		\end{align*}
		
		This proves the induction. In particular, by taking the logarithm of this expression for $k = m-1$, we get:
		\begin{align*}
			\log^{(m)}(\tow(m,m))
			& \leqslant 1 + \log m + \log \log m \\
			& \leqslant m \hspace{4cm} { \footnotesize \text{(because $m \geqslant 4$)} }
		\end{align*}
		
		Finally, since~$\log^\ast(m) \leqslant m$, we get for every integer~$m \geqslant 4$:
		\begin{align*}
			\log^\ast(\tow(m,m))
			& \leqslant m + \log^\ast(\log^{(m)}(\tow(m,m))) \\
			& \leqslant m + \log^\ast(m) \\
			& \leqslant 2m
		\end{align*}
		
		Now, let $n \geqslant \tow(4,4)$, and let $m := \dlog(n)$. Note that $m \geqslant 4$, so we have:
		\begin{align*}
			\log^\ast(n)
			& \leqslant \log^\ast(\tow(m,m)) \\
			& \leqslant 2m \\
			& \leqslant 2\dlog(n)
		\end{align*}
		
		This finally proves that $\log^\ast(n) = O(\dlog(n))$.
	\end{proof}

	\begin{proof}[Proof of Theorem~\ref{thm:explicit_lower_bound_parity}.]
		Assume that $o(\log^\ast n)$ different certificates are sufficient to certify an even number of vertices. For every integer~$m$, let~$n := m + 2^m - 1$.
		
		The proof is the same as the proof of Theorem~\ref{thm:parity lower bound}, but we just need to show that there exists
		an integer~$m$ such that $m \geqslant f(2,c)$, where~$c$ is the number of different certificates (which is a function of~$n$, so a function of~$m$).
		By assumption, we have~\mbox{$c = o(\log^\ast n)$}.
		Since $\log^\ast 2^m = 1 + \log^\ast m$, we have $\log^\ast n = \Theta(\log^\ast m)$. Thus, we also have~$c = o(\log^\ast m)$. By Lemma~\ref{lem:log*}, we have~$c = o(\dlog m)$.
		
		By Theorem~\ref{thm:bound_finite_unions}, we have~$f(2,c) \leqslant \tow(2c,2c)$. Let $m$ be large enough, such that~$c \leqslant (\dlog m - 1)/2$.
		Then, $f(2,c) \leqslant \tow(\dlog m -1, \dlog m -1) < m$.
		The rest of the proof is identical as the proof of Theorem~\ref{thm:parity lower bound}.
	\end{proof}

\begin{remark}
	The lower bound of Theorem~\ref{thm:explicit_lower_bound_parity} still holds in the following classes of graphs:
	\begin{enumerate}[(i)]
		\item bipartite graphs; \label{item:LB_even_bipartite}
		\item chordal graphs. \label{item:LB_even_chordal}
	\end{enumerate}
	Indeed, (\ref{item:LB_even_bipartite}) follows from the fact that graphs~$G_m$ and~$G_m'$ involved in the proof of Theorem~\ref{thm:parity lower bound} are bipartite. To prove~(\ref{item:LB_even_chordal}), we just need to change a bit the definition of~$G_m$ in the proof of Theorem~\ref{thm:parity lower bound}, by adding all the edges between vertices of~$X$. With this modification, the graph~$G_m$ becomes a split graph (that is, a graph whose vertices can be partitioned into a clique and an independent set) which is in particular chordal (that is, without any induced cycle of length at least~$4$). The proof still works identically, using the fact that gluing two chordal graphs by identifying one vertex in each gives another chordal graph.
\end{remark}

\vfill{}

	\textbf{AI Disclosure:} We used an LLM to suggest alternative formulations in the introduction. The authors verified the correctness and originality of all content including references.

\newpage{}

\bibliographystyle{plainurl}
\bibliography{biblio}

\end{document}